%% file: ms.tex
\RequirePackage{amsmath}

\documentclass[envcountsame]{llncs}
\usepackage[llncs]{optional}

\input{head}

\opt{llncs}{ \input{head_llncs} }
\opt{ieee}{ \input{head_ieeetran} }
\input{notation}

\title{
  Programming Substrate-Independent Kinetic Barriers
  with Thermodynamic Binding Networks}

\author{
  Keenan Breik\inst{1} \and
  Cameron Chalk\inst{1} \and
  David Doty\inst{2} \and
  David Haley\inst{2} \and
  David Soloveichik\inst{1}
}

\institute{
  University of Texas at Austin \and
  University of California, Davis
}

\begin{document}

\opt{llncs}{
\maketitle
\begin{abstract}
  {\input{abstract}}
\end{abstract}
}

\opt{ieee}{
\IEEEtitleabstractindextext{%
{\input{abstract}}
}
\maketitle
}

{\input{intro}}
{\input{model}}
{\input{saturated}}
{\input{modules}}
{\input{translator}}
{\input{hashgate}}
{\input{simulate}}
{\input{conclusion}}

\opt{ieee}{
\appendices
{\input{translator_appendix}}
}

\section*{Acknowledgements}
DD and DH were supported by NSF grants CCF-1619343, CCF-1844976, and CCF-1900931.
CC, KB, and DS were supported by NSF grants CCF-1618895 and CCF-1652824, and DARPA grant W911NF-18-2-0032.

\bibliography{chemcomp,big}
\bibliographystyle{plain}

\opt{llncs}{
\appendix
{\input{translator_appendix}}
}

\opt{ieee}{
{\input{biography}}
}

\clearpage

\end{document}

%% file: head.tex
\usepackage[utf8]{inputenc}
\usepackage[T1]{fontenc}
\usepackage{lmodern}
\usepackage{microtype}

\usepackage{enumitem}
\setlist{listparindent=\parindent}

\usepackage{xcolor}
\usepackage{graphicx}
\graphicspath{{}}

\usepackage{amsmath}
\usepackage{amssymb}
\usepackage{mathtools}
\mathtoolsset{showonlyrefs}
\usepackage{footmisc}

\usepackage[textsize=scriptsize]{todonotes}
\usepackage{letltxmacro}
\LetLtxMacro{\todom}{\todo}
\renewcommand{\todo}[1]{\todom[inline]{#1}}

\usepackage{verbatim}

\usepackage{float}

\usepackage{cite}

\usepackage[normalem]{ulem}

%% file: head_llncs.tex
\newcommand{\qedhere}{\qed}

\numberwithin{theorem}{section}

\spnewtheorem{lem}[theorem]{Lemma}{\bfseries}{}
\spnewtheorem{cor}[theorem]{Corollary}{\bfseries}{}
\spnewtheorem{defn}[theorem]{Definition}{\bfseries}{}
\spnewtheorem{observation}[theorem]{Observation}{\bfseries}{}

%% file: notation.tex
\newcommand{\HL}[1]{#1}%

\newcommand{\placeholder}{\HL{\;\!{\cdot}\;\!}}
\newcommand{\markdef}[1]{}%
\makeatletter
\newcommand{\defterm}[2][\@nil]{%
  \HL{\emph{#2}}%
  \def\formargin{#1}%
  \ifx\formargin\@nnil
    \def\formargin{#2}%
  \fi
  }
\makeatother

\DeclarePairedDelimiter\abs{\lvert}{\rvert}

\newcommand{\defeq}{=}%

\newcommand{\figlabel}[1]{\label{fig:#1}}
\newcommand{\figref}[1]{\HL{Fig.~\ref{fig:#1}}}%
\newcommand{\seclabel}[1]{\label{sec:#1}}
\newcommand{\secref}[1]{\HL{Section~\ref{sec:#1}}}
\newcommand{\eqlabel}[1]{\label{eqx:#1}}
\let\oldeqref\eqref
\renewcommand{\eqref}[1]{\HL{\oldeqref{eqx:#1}}}

\newcommand{\site}[1]{\HL{\ifcase#1\or a\or b\or c\else\@ctrerr\fi}}
\newcommand{\mono}[1]{\HL{\ifcase#1\or m\or n\else\@ctrerr\fi}}
\newcommand{\poly}[1]{\HL{\ifcase#1\or P\or Q\else\@ctrerr\fi}}
\newcommand{\conf}[1]{\HL{\ifcase#1%
  \or \gamma\or \delta\or \alpha\or \beta\else\@ctrerr\fi}}

\newcommand{\tbn}[1]{\HL{\mathcal{\ifcase#1\or T\or V\else\@ctrerr\fi}}}
\newcommand{\rec}[1]{\HL{\ifcase#1\or r\or s\else\@ctrerr\fi}}
\newcommand{\path}[1]{\HL{\ifcase#1\or p\or q\else\@ctrerr\fi}}

\newcommand{\flip}[1]{#1^{\HL{*}}}
\newcommand{\bondworth}{\HL{w}}
\newcommand{\bondcount}[1]{\HL{H(}#1\HL{)}}
\newcommand{\polycount}[1]{\HL{S(}#1\HL{)}}
\newcommand{\energy}[1]{\HL{E(}#1\HL{)}}
\newcommand{\pathcost}[1]{\HL{h(}#1\HL{)}}
\newcommand{\barrier}[3][]{\HL{b}_{#1}\HL{(}#2\HL{,} #3\HL{)}}
\newcommand{\satbarrier}[2]{\barrier[\text{sat}]{#1}{#2}}
\newcommand{\bondbarrier}[2]{\barrier[\text{bond}]{#1}{#2}}
\newcommand{\satbondbarrier}[2]{\barrier[\text{sat-bond}]{#1}{#2}}

\newcommand{\mix}[2]{#2_{#1}}%
\newcommand{\mrg}[1]{\mix{#1}{\succ}}%
\newcommand{\mrgeq}[1]{\mix{#1}{\succeq}}%
\newcommand{\splt}[1]{\mix{#1}{\prec}}%
\newcommand{\splteq}[1]{\mix{#1}{\preceq}}%
\newcommand{\sanymerge}{\mrg{}}%
\newcommand{\anymerge}{\mrgeq{}}%
\newcommand{\sanysplit}{\splteq{}}%
\newcommand{\anysplit}{\splteq{}}%
  \newcommand{\sxanymerge}{\sanymerge^1}%
  \newcommand{\xanymerge}{\anymerge^1}%
  \newcommand{\sxanysplit}{\sanysplit^1}%
\newcommand{\scleanmerge}{\mrg{\circ}}%
\newcommand{\cleanmerge}{\mrgeq{\circ}}%
\newcommand{\scleansplit}{\splt{\circ}}%
\newcommand{\cleansplit}{\splteq{\circ}}%
  \newcommand{\sxcleanmerge}{\scleanmerge^1}%
  \newcommand{\xcleanmerge}{\cleanmerge^1}%
  \newcommand{\sxcleansplit}{\scleansplit^1}%
\newcommand{\sbindmerge}{\mrg{\bullet}}%
\newcommand{\bindmerge}{\mrgeq{\bullet}}%
  \newcommand{\sxbindmerge}{\sbindmerge^1}%

\newcommand{\satbindmerge}[1]{\HL{[}#1\HL{]}}%

\newcommand{\simp}[1]{\HL{\langle} #1 \HL{\rangle}}%

\newcommand{\init}{\HL{\conf1_{I}}}%
\newcommand{\trig}{\HL{\conf1_{T}}}%

\newcommand{\strandlength}{\HL{z}}%
\newcommand{\fuelcount}{\HL{c}}%

\newcommand{\hashgate}{grid gate}
\newcommand{\Hashgate}{Grid gate}
\newcommand{\Translatorcycle}{Translator cycle}
\newcommand{\translatorcycle}{translator cycle}

%% file: abstract.tex
Engineering molecular systems that exhibit complex behavior requires the design of kinetic barriers.
For example, an effective catalytic pathway must have a large barrier when the catalyst is absent.
While programming such energy barriers seems to require knowledge of the specific molecular substrate, we develop a novel substrate-independent approach.
We extend the recently-developed model known as thermodynamic binding networks, demonstrating programmable kinetic barriers that arise solely from the thermodynamic driving forces of bond formation and the configurational entropy of forming separate complexes.
Our kinetic model makes relatively weak assumptions, which implies that energy barriers predicted by our model would exist in a wide variety of systems and conditions.
We demonstrate that our model is robust by showing that several variations in its definition result in equivalent energy barriers.
We apply this model to design catalytic systems with an arbitrarily large energy barrier to uncatalyzed reactions.
Our results could yield robust amplifiers using DNA strand displacement, 
a popular technology for engineering synthetic reaction pathways, and suggest design strategies for preventing undesired kinetic behavior in a variety of molecular systems.

%% file: intro.tex
\section{Introduction}
Abstract mathematical models of molecular systems,
such as chemical reaction networks,
have long been useful in \emph{natural} science
to study the properties of natural molecules.
With recent experimental advances in synthetic biology and DNA nanotechnology~\cite{soloveichik2010dna,cardelli2011strand,chen2013programmable,baccouche2014dynamic,srinivas2017enzyme},
such models have come to be viewed also as \emph{programming languages} 
for describing the desired behavior of synthetic molecules.

We can describe a chemical program
with abstract chemical reactions
such as
\begin{align}
  A + C &\to B + C
  \eqlabel{cat}
  \\
  A &\to B
  \eqlabel{uncat}
  .
\end{align}
In particular,
a program may require \eqref{cat}
and forbid \eqref{uncat}.
But what remains hidden at this level of abstraction
is a well-known chemical constraint:
if \eqref{cat} is possible,
then \eqref{uncat} must also be,
no matter the exact substances.
Knowing this,
we might try to slow \eqref{uncat}
by ensuring $B$ has high free energy.
But then $B + C$ must also have high free energy,
so \eqref{cat} slows in tandem.
The only option
to slow \eqref{uncat} but not \eqref{cat}
is to use a \emph{kinetic barrier}:
designing $A$ so that, although it is possible for $A$ to reconfigure into $B$,
the system must traverse a higher energy (less favorable) intermediate in the absence of $C$.

It seems difficult to engineer kinetic energy barriers and catalysis in a way that is independent of the particular chemical substrate. 
For example, the development of novel protein enzymes requires the precise positioning of hydrophobic or electrostatic interactions, or otherwise chemically active sites, which are often hard to engineer from first principles.
Catalysts based on DNA strand displacement reactions
arguably promise the highest degree of programmability~\cite{ZhaTurYurWin07,yin2008programming}.
Yet, kinetic barriers here usually crucially depend on the  specific toehold-sequestering mechanism.
Importantly, such DNA-based catalysis suffer from a relatively large rate of the uncatalyzed reaction, and 
prior work lacks a method to arbitrarily increase the uncatalyzed energy barrier.

To develop a substrate-independent approach to engineering kinetic barriers we need to rely on a universal thermodynamic property that would be relevant in a wide variety of chemical systems.
We focus on the entropic penalty of association (decreasing the number of separate complexes). 
Intuitively, the entropic penalty is due to decreasing the number of microstates corresponding to the independent three-dimensional positions of each complex (configurational entropy).
This thermodynamic penalty can be made dominant compared with other factors by decreasing the concentration.

To formalize this entropic penalty, we use the \emph{thermodynamic binding networks} (TBN) model~\cite{tbn,breik2019computing}.
TBNs represent molecules as abstract \emph{monomers} with binding sites that allow them to bind to other monomers.
For a configuration $\gamma$, the TBN model defines $H(\gamma)$ as the number of bonds formed,
and $S(\gamma)$ as the number of free complexes,
and the \emph{energy} 
$E(\gamma) = - \bondworth H(\gamma) - S(\gamma)$ as a (negative) weighted sum of the two.
To be applicable to a wide variety of chemical systems, the TBN model does not impose geometric constraints on bonding (monomers are simply multisets of binding sets).
Implementation of TBNs requires choosing a concrete physical substrate and geometric arrangement that permits the desired configurations to form.

We augment the TBN model with a notion of kinetic paths (changes in configuration) due to merging of different complexes and splitting them up (and in \secref{bonds} making, breaking, or exchanging bonds).
This gives rise to a notion of \emph{paths} of configurations,  with different energies.
Define the \emph{height} of a path starting at $\conf1$ 
as the maximum energy difference $E(\conf2) - E(\conf1)$ over all configurations $\conf2$ on the path.
Then the kinetic energy \emph{barrier} separating configuration $\conf2$ from configuration $\conf1$ 
is the height of the minimum-height path from $\conf1$ to $\conf2$.

In \secref{model} we introduce our main kinetic model.
In \secref{saturated} we give a sufficient condition for large kinetic barriers.
In \secref{modules}
we develop two constructions for catalytic systems.
Both constructions yield families of TBNs parametrized by a complexity parameter $n$ such that the uncatalyzed energy barrier scales linearly with $n$. 
The catalyzed energy barrier is always $1$.
We show a direct DNA strand displacement implementation of one of the constructions.
Finally we show an autocatalytic TBN, with an arbitrarily large energy barrier to undesired triggering, that exponentially amplifies its input signal (Section~\ref{hashgate:catsection}).
In \secref{bonds}
we show that our main model predicts the same barriers
as a more complicated model,
one that keeps track of individual bonds
on top of whole polymers.

%% file: model.tex
\setcounter{secnumdepth}{2}

\section{Kinetic model} \seclabel{model}

{\input{figtbnconf}}

Our kinetic models build on thermodynamic binding networks (TBN)~\cite{tbn,breik2019computing}.
Intuitively, we model a chemical system as a collection of molecules,
each of which has a collection of binding sites,
which can bind if they are complementary.
Although the TBN model is more general,
DNA domains can be thought of as the prototypical example of binding sites.
No geometry is enforced,
which allows the model to handle topologically complex structures,
such as pseudoknots.

\subsection{TBN}

(See \figref{tbnconf}.)
Formally, a \defterm{TBN} is a multiset of monomer types.
A \defterm{monomer type} is a multiset of site types.
A \defterm{site type} is a formal symbol, such as $\site1$,
and has a \defterm{complementary} type, denoted $\flip{\site1}$%
\markdef{$\flip{}$}.
We call an instance of a monomer type a \defterm{monomer}
and an instance of a site type a \defterm{site}.

\subsection{Configuration}

(See \figref{tbnconf}.)
We may describe the configuration of a TBN at any moment
in terms of which monomers are grouped into polymers.
This way a \defterm{polymer} is a multiset of monomer types,
and a \defterm{configuration} is a partition of the TBN into polymers.%
\footnote{
Consider swapping two monomers of the same type between two polymers in a configuration.
We do not consider the result a different configuration.
Note that monomers of the same type
correspond to an entropy contribution that we ignore (see also footnote \ref{foot:entropies}).
}

The \defterm[exposed site]{exposed sites} of a polymer
is the multiset of site types
that would remain if one were to remove
as many complementary pairs of sites as possible.
Each such pair is counted as a \defterm{bond}.
Note that bonds are not specified as part of a configuration, and intuitively we think of polymers as being maximally bonded.
Two polymers are \defterm{compatible}
if they have some complementary exposed sites.
A configuration is \defterm{saturated}
if no two polymers are compatible.
This is equivalent to having the maximum possible number of bonds.

\subsection{Path}

{
\newcommand{\fine}{\conf1}%
\newcommand{\coarse}{\conf2}%

(See \figref{pathenergy}.)
One configuration can change into another
by a sequence of elementary steps.
If $\fine$ can become $\coarse$
by replacing two polymers in $\fine$ with their union,
then $\fine$ \defterm[merge]{merges} to $\coarse$
and $\coarse$ \defterm[split]{splits} to $\fine$,
and we write $\fine \sxanymerge \coarse$%
\markdef{$\sxanymerge$}.%
\footnote{Note that this would form a lattice of partitions if the configurations were  sets instead of multisets.}
We denote by ${\xanymerge}$, ${\sanymerge}$, ${\anymerge}$
the reflexive, transitive, and reflexive transitive closures of ${\sxanymerge}$.
A \defterm{path} is a nonempty sequence of configurations
where each merges or splits to the next.
Note that there is a path between any two configurations.\footnote{
    Our model allows incompatible polymers to be merged
    (i.e., two polymers merge without forming any new bonds).
    This represents spontaneous co-localization
    and comes with an energy penalty,
    as discussed later. 
    For instance, to get from any configuration to any other, we can merge all initial polymers into one 
    and then split into the desired end polymers; however, 
    such a path could be very energetically unfavorable.
}

We could imagine smaller steps
that manipulate individual bonds.
But surprisingly,
such a bond-aware model
leads to essentially equivalent kinetic barriers,
which we prove in \secref{bonds}.
}

{\input{figpathenergy}}

\subsection{Energy}

(See \figref{pathenergy}.)
For a configuration $\conf1$,
denote by $\bondcount{\conf1}$%
\markdef{$\bondcount{\placeholder}$}
the number of bonds summed over all polymers.
Denote by $\polycount{\conf1}$%
\markdef{$\polycount{\placeholder}$}
the number of polymers.%
\footnote{The quantities $H(\gamma)$ and $S(\gamma)$ are meant to evoke the thermodynamic quantities of enthalpy and entropy, although the mapping is not exact.
Indeed, the free energy contribution of forming additional bonds typically contains substantial enthalpic and entropic parts.
Further, while $\polycount{\conf1}$ captures entropy due to independent positions of separate polymers,
chemical free energy may consider a variety of other entropic contributions.
These may include geometric configurations of a single polymer, as well as the entropy due to swapping indistinguishable monomers.
We can justify focusing on $S(\conf1)$ because its contribution arbitrarily predominates in taking the limit of large solution volume---that is, the low concentration regime~\cite{tbn}.
\label{foot:entropies}
}
Note that a saturated configuration
has maximum $\bondcount{\conf1}$.
The \defterm{energy}
of $\conf1$ is
\begin{equation}
  \energy{\conf1} = - \bondworth \bondcount{\conf1} - \polycount{\conf1},
  \markdef{$\energy{\placeholder}$}
\end{equation}
where the bond strength $\bondworth$%
\markdef{$\bondworth$}
represents the benefit from gaining a bond
relative to gaining a polymer.%
\footnote{%
Our notion of energy idealizes the physical Gibbs free energy.
In typical DNA nanotechnology applications, 
the Gibbs free energy is likewise a linear combination of 
$H(\gamma)$ and $S(\gamma)$.
We can estimate the Gibbs free energy $\Delta G(\gamma)$ of a configuration $\gamma$ as follows.
Bonds correspond to domains of length~$l$ bases,
and forming each base pair is favorable by $\Delta G^\circ_{\text{bp}}$.
Thus, the contribution of $H(\gamma)$ to $\Delta G(\gamma)$ is $(\Delta G^\circ_{\text{bp}} \cdot l) H(\gamma)$.
At $1$ M concentration, the free energy penalty due to decreasing the number of separate complexes by one is $\Delta G^\circ_{\text{assoc}}$.
At lower concentration $c \text{ M} < 1 \text{ M}$, this penalty increases to 
$\Delta G^\circ_{\text{assoc}} + RT \ln((1\text{ M})/c)$.
The point of zero free energy is taken to be the configuration with no bonds, and all monomers separate.
Thus, the contribution of $S(\gamma)$ to $\Delta G(\gamma)$ is $(\Delta G^\circ_{\text{assoc}} + RT \ln((1\text{ M})/c))(m-S(\gamma))$, where $m$ is the total number of monomers.
To summarize,
\begin{align}
\Delta G(\gamma) &= 
(\Delta G^\circ_{\text{bp}} \cdot l) H(\gamma) + 
(\Delta G^\circ_{\text{assoc}} + RT \ln((1\text{ M})/c))(m-S(\gamma)).
\end{align}
Note that, as expected, this is a linear combination of $H(\gamma)$ and $S(\gamma)$, and that increasing the length of domains~$l$ weighs $H(\gamma)$ more heavily, while decreasing the concentration $c$ weighs $S(\gamma)$ more heavily.
Domains are routinely $15-25$ bases long, and at $100$ nM concentration at room temperature this corresponds to a relative bond strength $w$ of $1.9$--$3.2$.
\label{foot:gibbs}}
Note that $\bondcount{\conf1} \geq 0$ and $\polycount{\conf1} > 0$,
so $\energy{\conf1} < 0$,
and that lower energy,
which results from more bonds or more polymers,
is more favorable.
(The choice to make favorability correspond to lower energy, that is more negative, is motivated by consistency with the standard physical chemistry notion of free energy.)
We call a minimum energy configuration \defterm{stable}.

Merging incompatible polymers
forms no additional bonds
and so is unfavorable,
since
$\polycount{\conf1}$ drops without $\bondcount{\conf1}$ rising.
In contrast, when bond strength $\bondworth > 1$,
merges between compatible polymers are energetically favorable.
So every stable (that is, minimum energy) configuration is saturated.
This regime is typical of many real systems, and 
in particular, we can engineer DNA strand displacement systems~\cite{thachuk2015leakless}
to have large bond strength $\bondworth$ by increasing the length of domains.

\subsection{Barrier}

{
\newcommand{\start}{\conf1}%
\newcommand{\finish}{\conf2}%
\newcommand{\along}{\conf3}%

(See \figref{pathenergy}.)
There are many paths from a start configuration $\start$
to an end configuration $\finish$.
The \defterm{height} $\pathcost{\path1}$%
\markdef{$\pathcost{\placeholder}$}
of a particular such path $\path1$
is the greatest energy difference
$\energy{\along} - \energy{\start}$
between any $\along$ along $\path1$ and $\start$. 
This measures the difficulty of traversing $\path1$.
Notice that $\pathcost{\path1} \geq \energy{\start} - \energy{\start} = 0$. 

Another reasonable definition of height is the greatest energy difference $\energy{\beta} - \energy{\alpha}$
between any $\alpha$ and later $\beta$ on the path.
We will be considering paths between stable (lowest energy) configurations, where the two definitions are equivalent.

Going from one configuration to another is difficult
if each path has large height.
The \defterm{barrier}
$\barrier{\start}{\finish}$%
\markdef{$\barrier{\placeholder}{\placeholder}$}
from $\start$ to $\finish$
is the least height of any path from $\start$ to $\finish$.
Notice that $\barrier{\start}{\finish} \geq 0$ as well.
Since paths are reversible, it is easy to show that if $\energy{\start} = \energy{\finish}$ then  $\barrier{\start}{\finish} = \barrier{\finish}{\start}$.
}

%% file: figtbnconf.tex
\begin{figure}[t]
  \centering
  \opt{llncs}{\includegraphics[width=0.76\columnwidth]{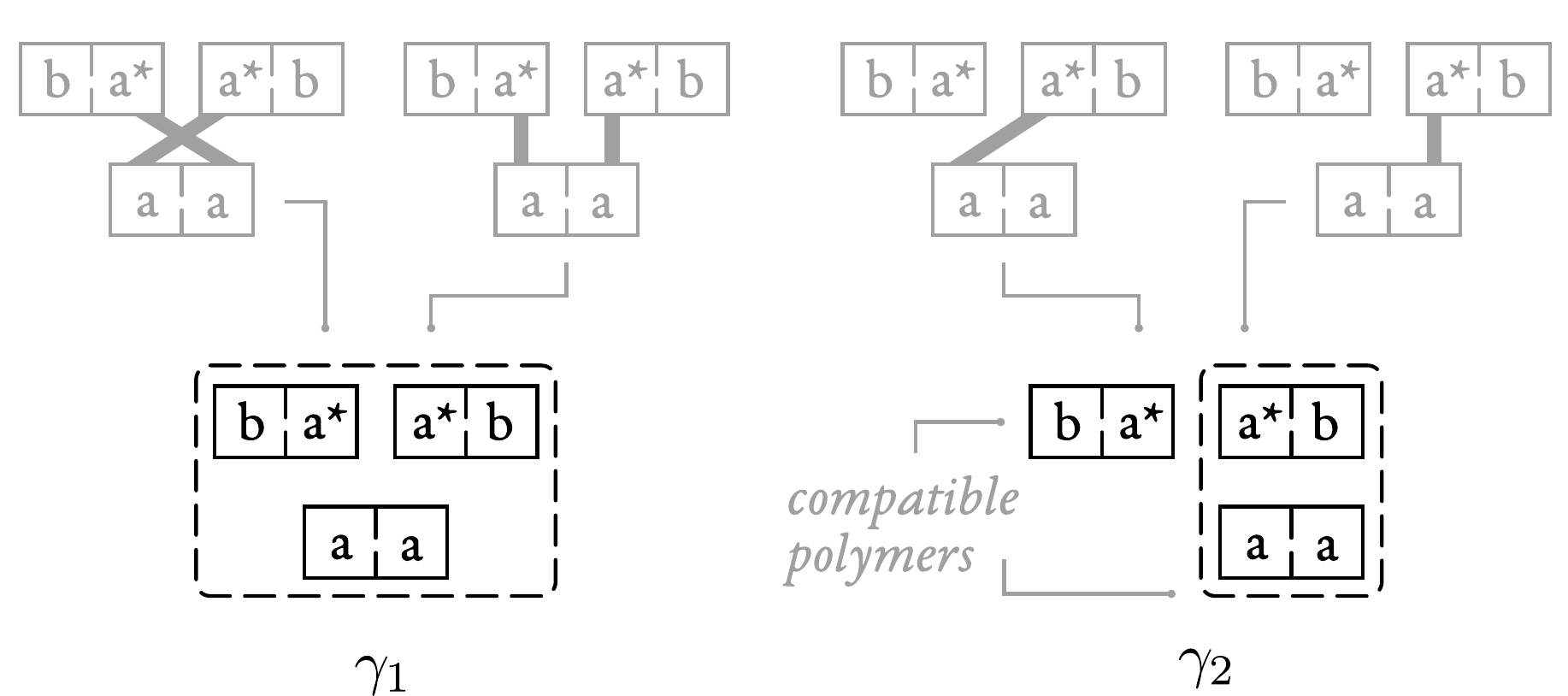}}
  \opt{ieee}{\includegraphics[width=1.0\columnwidth]{tbnconf}}
  \caption{
    Two configurations $\gamma_1$ and $\gamma_2$ of the TBN
    $\tbn1 = \{ \{a, a\}, \{\flip{a}, b\}, \{\flip{a}, b\} \}$.
    Note that $\tbn1$ has 3 monomers but 2 monomer types
    and 6 sites but 3 site types.  
    A dashed box indicates monomers that are part of the same polymer.
    A single configuration (bottom) can correspond
    to multiple ways of binding complementary sites (top),
    which are not distinguished in our model.
    In $\gamma_2$ the polymer on the left has exposed sites $\{b,\flip{a}\}$ and the polymer on the right $\{a, b\}$; they are thus compatible since the exposed site $\flip{a}$ of the left is complementary to exposed site $a$ of the right.
    Since $\gamma_2$ has compatible polymers it is not saturated, but $\gamma_1$ is.
  }
  \figlabel{tbnconf}
\end{figure}

%% file: figpathenergy.tex
\begin{figure}[t]
  \centering
  \opt{llncs}{\includegraphics[width=0.82\columnwidth]{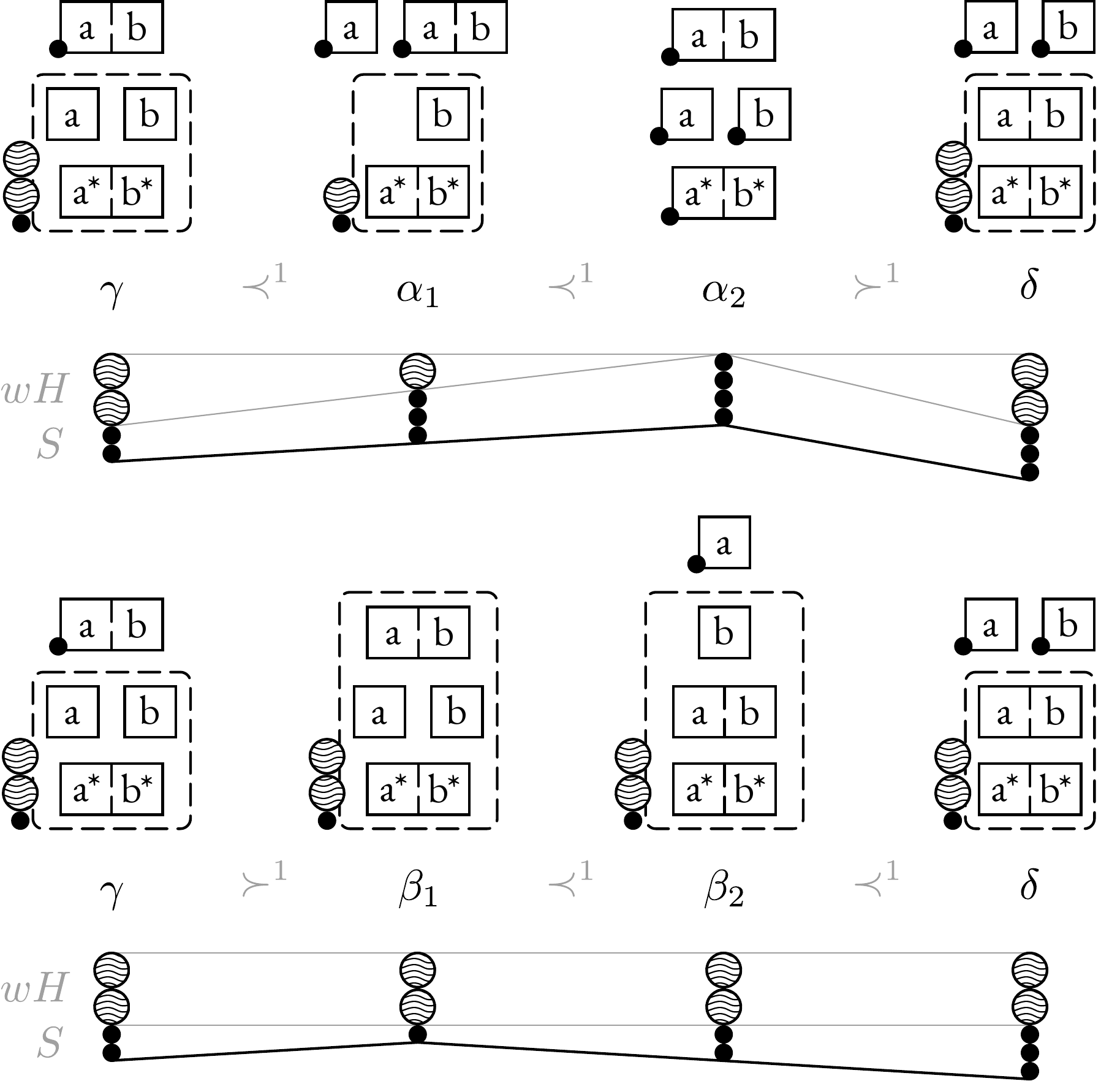}}
  \opt{ieee}{\includegraphics[width=1.0\columnwidth]{pathenergy}}
  \newcommand{\pathtop}{\path1}%
  \newcommand{\pathbot}{\path1'}%
  \caption{%
    A path $\pathtop$
    consisting of
    $\gamma$, $\alpha_1$, $\alpha_2$, $\delta$
    and a path $\pathbot$
    consisting of
    $\gamma$, $\beta_1$, $\beta_2$, $\delta$
    in the TBN
    $\tbn1 = \{\{a\}, \{b\}, \{a, b\}, \{\flip{a}, \flip{b}\}\}$.
    The energy of each configuration is shown graphically below it.
    A large wavy disc indicates energy due to a bond.
    A small solid disc indicates energy due to a polymer.
    Lower is more favorable.
    Here bond strength $\bondworth = 2$,
    so a wavy disc is twice as tall as a solid disc.
    The height of $\pathtop$ is
    $\pathcost{\pathtop}
    = \energy{\alpha_2} - \energy{\gamma}
    = (-4) - (-2 \bondworth - 2)
    = 2$.
    The height of $\pathbot$ is
    $\pathcost{\pathbot}
    = \energy{\beta_1} - \energy{\gamma}
    = (-2 \bondworth - 1) - (-2 \bondworth - 2)
    = 1$.
  }
  \figlabel{pathenergy}
\end{figure}

%% file: saturated.tex
\section{Saturated paths} \seclabel{saturated}

Analyzing TBNs is simpler if we reason only about saturated configurations.
Clearly, in the limit of large bond strength $w$, breaking a bond is so unfavorable that a least height path has only saturated configurations.
The main result of this section, 
Corollary~\ref{cor:satbarrieralphatwo},
shows that the barrier remains the same
even if we consider paths that traverse only saturated configurations
as long as $\bondworth \geq 2$.
Such a threshold may be surprising:
it might seem that breaking some bonds, even if locally unfavorable, 
might allow a path to bypass an otherwise large barrier elsewhere.

We prove Corollary~\ref{cor:satbarrieralphatwo} in \secref{sat-paths-suffice} by showing how to transform an arbitrary path into a saturated path with height no greater. \secref{bounds-energy-change} defines a relation among configurations that allows us to keep track of changes in energy as we transform the path.

\subsection{Bounds on energy change}
\seclabel{bounds-energy-change}

{
\newcommand{\fine}{\conf1}%
\newcommand{\coarse}{\conf2}%

When two polymers merge,
knowing whether they are compatible
makes the change in energy predictable.
Recall that merging incompatible polymers results in no more bonds,
so overall energy increases by 1.
Merging compatible polymers results in at least one more bond,
so overall energy decreases by at least $\bondworth - 1$.
To make this precise,
let $\fine \sxcleanmerge \coarse$
(and let $\fine \sxbindmerge \coarse$)
mean that $\fine$ merges to $\coarse$
by combining two incompatible (compatible) polymers.
Let $\cleanmerge$ ($\bindmerge$)
be the reflexive, transitive closure of $\sxcleanmerge$
($\sxbindmerge$).

\begin{observation}
  If $\fine \sxcleanmerge \coarse$,
  then $\energy{\coarse} = \energy{\fine} + 1$.
  If $\fine \sxbindmerge \coarse$,
  then $\energy{\coarse}
  \leq \energy{\fine} + 1 - \bondworth$.
\end{observation}

\begin{observation}
  \newcommand{\npolyslost}{\Delta}%
  Let $\npolyslost = \polycount{\fine} - \polycount{\coarse}$.
  If $\fine \cleanmerge \coarse$,
  then $\energy{\coarse} = \energy{\fine} + \npolyslost$.
  If $\fine \bindmerge \coarse$,
  then $\energy{\coarse}
  \leq \energy{\fine} + \npolyslost (1 - \bondworth)$.
\end{observation}

\noindent
To apply these bounds to the general case $\fine \anymerge \coarse$,
we decompose $\anymerge$ into $\bindmerge$ and $\cleanmerge$.
This allows us to identify an intermediate configuration
that has least energy (is most favorable).

\newcommand{\medium}{\conf3}%

\begin{theorem} \label{thm:canbindfirst}
  If $\fine \anymerge \coarse$,
  then some $\medium$ has $\fine \bindmerge \medium \cleanmerge \coarse$.
\end{theorem}

\begin{proof}
  \newcommand{\step}[1]{\conf4_{#1}}%
  \newcommand{\wholea}{\poly1}%
  \newcommand{\wholeb}{\poly2}%
  \newcommand{\arb}{k}%

  Let $\fine \anymerge \coarse$,
  and let $\medium$ be a 
  ``most merged'' configuration
  with $\fine \bindmerge \medium \anymerge \coarse$
  (no $\medium' \sanysplit \medium$
  has $\fine \bindmerge \medium' \anymerge \coarse$).
  Then $\medium = \step0 \sxanymerge \cdots \sxanymerge \step{n} = \coarse$
  for some configurations $\step{i}$.
  Consider the polymers $\wholea$ and $\wholeb$ in $\step{\arb}$
  merged by $\step{\arb} \sxanymerge \step{\arb + 1}$.
  We claim $\wholea$ and $\wholeb$ are incompatible.

  \newcommand{\parta}{\wholea}%
  \newcommand{\partb}{\wholeb}%
  \newcommand{\exposed}[1]{[#1]}%
  
  To see so,
  note that $\wholea = \parta_1 \cup \cdots \cup \parta_x$
  and $\wholeb = \partb_1 \cup \cdots \cup \partb_y$
  for some polymers $\parta_i$ and $\partb_j$ in $\medium$.
  If any $\parta_i$ and $\partb_j$ are compatible,
  then merging them in $\medium$ would produce a configuration
  that contradicts $\medium$ being most merged.
  So they are pairwise incompatible.
  Letting $\exposed{X}$ denote the exposed sites of $X$,
  we have $\exposed{\wholea}
  \subseteq \exposed{\parta_1} \cup \cdots \cup \exposed{\parta_x}$
  and $\exposed{\wholeb}
  \subseteq \exposed{\partb_1} \cup \cdots \cup \exposed{\partb_y}$.
  So $\wholea$ and $\wholeb$ are incompatible.
  
  So $\step{\arb} \sxcleanmerge \step{\arb + 1}$
  for each $\arb$,
  and $\medium = \step0 \cleanmerge \step{n} = \coarse$.
  \qedhere
\end{proof}
}

\subsection{Saturated paths suffice}
\seclabel{sat-paths-suffice}

\newcommand{\orig}{\path1}%
\newcommand{\sat}{\path1'}%
\newcommand{\start}{\conf1}%
\newcommand{\finish}{\conf2}%
\newcommand{\startsat}{\start'}%
\newcommand{\finishsat}{\finish'}%

A \defterm{saturated path}
is a path along which every configuration is saturated.
For example, the bottom path $\path1'$ in \figref{pathenergy} is saturated.
If $\conf1$ and $\conf2$ are saturated,
then let $\satbarrier{\conf1}{\conf2}$%
\markdef{$\satbarrier{\placeholder}{\placeholder}$}
denote the barrier from $\conf1$ to $\conf2$
when allowing only saturated paths.
Since a saturated path is a path,
$\satbarrier{\start}{\finish} \geq
\barrier[]{\start}{\finish}$.
It turns out that
if bond strength $\bondworth \geq 2$,
then the reverse inequality also holds,
so $\satbarrier{\start}{\finish} = \barrier{\start}{\finish}$.
And if $\bondworth \geq 1$,
then the reverse inequality ``almost'' holds.

To show the reverse inequality,
we turn an arbitrary path
into a saturated path
without increasing its height (much).
We do this step by step
by always merging just enough polymers
to achieve saturation.
To make this precise,
let $\satbindmerge{\conf1}$%
\markdef{$\satbindmerge{\placeholder}$}
denote the set of saturated $\conf1'$
with $\conf1 \bindmerge \conf1'$,
\newcommand{\maxenergy}[1]{E_{\max}(#1)}%
and let $\maxenergy{\path1}$%
\markdef{$\maxenergy{\placeholder}$}
be the maximum energy
of any configuration along the path $\path1$.

First we show how to saturate a split step.

\begin{lem} \label{lem:satsplit}
  Let $\bondworth \geq 1$.
  If $\start \sxanysplit \finish$ and $\startsat \in \satbindmerge{\start}$,
  then some $\finishsat \in \satbindmerge{\finish}$
  and some saturated path $\sat$ from $\startsat$ to $\finishsat$ has
  $%
    E_{\max}(\sat)
      \leq \energy{\start}
    .
  $%
\end{lem}

\begin{proof}
  Let $\start \sxanysplit \finish$
  and $\startsat \in \satbindmerge{\start}$.
  Then $\start \anymerge \startsat$,
  so $\finish \anymerge \startsat$.
  So by Theorem~\ref{thm:canbindfirst},
  some $\finishsat$ has
  $\finish \bindmerge \finishsat \cleanmerge \startsat$.
  By assumption, $\startsat$ is saturated,
  so $\finishsat$ is,
  so $\finishsat \in \satbindmerge{\finish}$.
  Now let $\sat$ be a path guaranteed by
  $\startsat \cleansplit \finishsat$.
  Then $E_{\max}(\sat) = \energy{\startsat} \leq \energy{\start}$.
  \qedhere
\end{proof}
\noindent
To show how to saturate a merge step,
we rely on being able to transfer a merge from one context to another.

{
\newcommand{\pre}{\start}%
\newcommand{\post}{\finish}%
\newcommand{\precoarse}{\smash{\pre'}\vphantom{\pre}}%
\newcommand{\postcoarse}{\conf3}%

\begin{lem} \label{lem:movemerge}
  If $\pre \sxanymerge \post$
  and $\pre \anymerge \precoarse$,
  then some $\postcoarse$ has
  $\precoarse \xanymerge \postcoarse$
  and $\post \anymerge \postcoarse$.
\end{lem}

\begin{proof}
  \newcommand{\blob}{\poly1}%
  \newcommand{\trans}[1]{#1 ^ {*}}%
  \newcommand{\step}[1]{\conf4 _ {#1}}%
  Let $\pre \sxanymerge \post$
  and $\pre \anymerge \precoarse$.
  Let $\blob$ be the polymer merged by $\pre \sxanymerge \post$,
  and let $\trans{\pre}$ be $\pre$ but with all polymers intersecting $\blob$ merged.
  This way $\trans{\pre} = \post$
  and $\precoarse \xanymerge \trans{\precoarse}$.

  Now $\pre = \step0 \sxanymerge \cdots \sxanymerge \step{n} = \precoarse$
  for some configurations $\step{i}$.
  So $\post = \trans{\pre} = \trans{\step0}
  \xanymerge \cdots
  \xanymerge \trans{\step{n}} = \trans{\precoarse}$.
  So choose $\postcoarse = \trans{\precoarse}$.
  \qedhere
\end{proof}
}
\noindent
Now we show how to saturate a merge step.

\begin{lem} \label{lem:satmerge}
  Let $\bondworth \geq 1$.
  If $\start \sxanymerge \finish$ and $\startsat \in \satbindmerge{\start}$,
  then some $\finishsat \in \satbindmerge{\finish}$
  and some saturated path $\sat$ from $\startsat$ to $\finishsat$ has
  $%
    E_{\max}(\sat) \leq
      \max \{ \energy{\start}, \energy{\finish} \}
        + \max \{ 0, 2 - \bondworth \}
    .
  $%
\end{lem}

\begin{proof}
  Let $\start \sxanymerge \finish$
  and $\startsat \in \satbindmerge{\start}$.
  If $\startsat = \start$,
  then let $\finishsat = \finish$,
  and let $\sat = \start, \finish$.
  Then $E_{\max}(\sat) = \energy{\finish}$.

  \newcommand{\medium}{\conf3}%

  Otherwise $\startsat \neq \start$.
  Now by Lemma~\ref{lem:movemerge},
  some $\medium$ has
  $\startsat \xanymerge \medium$
  and $\finish \anymerge \medium$.
  So by Theorem~\ref{thm:canbindfirst},
  some $\finishsat$ has
  $\finish \bindmerge \finishsat \cleanmerge \medium$.
  Since $\startsat$ is saturated,
  $\medium$ is,
  so $\finishsat$ is,
  so $\finishsat \in \satbindmerge{\finish}$.
  Now let $\sat$ be a path guaranteed by
  $\startsat \xcleanmerge \medium \cleansplit \finishsat$.
  Then $\sat$ is saturated.
  Also, $E_{\max}(\sat) = \energy{\medium}$.
  And $\startsat \xcleanmerge \medium$,
  so $\energy{\medium} \leq \energy{\startsat} + 1$.
  Since $\startsat \in \satbindmerge{\start}$
  and $\start \neq \startsat$,
  we have $\start \sbindmerge \startsat$,
  so $\energy{\startsat} \leq \energy{\start} + 1 - \bondworth$.
  So $E_{\max}(\sat) \leq \energy{\start} + 2 - \bondworth$,
  and the result follows from the identity $x \leq \max \{ x, y \}$.
  \qedhere
\end{proof}

\noindent
To saturate a full path,
we saturate each step.

\begin{theorem} \label{thm:satbarrier}
  Let bond strength $\bondworth \geq 1$
  and $\start$ and $\finish$ be saturated.
  Then
  \begin{equation}
    \satbarrier{\start}{\finish}
      \leq \barrier{\start}{\finish} + \max \{ 0, 2 - \bondworth \}
      .
  \end{equation}
\end{theorem}

\begin{proof}
  Let $\alpha_1$ and $\alpha_n$ be saturated.
  Consider a path $\orig = \alpha_1, \ldots, \alpha_n$.
  Let $\alpha_1' = \alpha_1$.
  Then $\alpha_1' \in \satbindmerge{\alpha_1}$.
  So by Lemmas~\ref{lem:satsplit} and~\ref{lem:satmerge},
  for each $i$ we get $\alpha_{i + 1}' \in \satbindmerge{\alpha_{i + 1}}$
  and saturated $\path1_i'$
  from $\alpha_i'$ to $\alpha_{i + 1}'$
  with $E_{\max}(\path1_i')
  \leq \max \{ \energy{\alpha_i}, \energy{\alpha_{i + 1}} \} + M_{\bondworth}$
  where $M_{\bondworth} = \max \{ 0, 2 - \bondworth \}$.
  Now $\alpha_n$ is saturated,
  so $\alpha_n' = \alpha_n$.
  So let $\sat$ be the concatenation of the $\path1_i'$.
  Then $\sat$ is a saturated path from $\alpha_1$ to $\alpha_n$.
  And we have
  \begin{align}
    E_{\max}(\sat)
      = {} & \begingroup \textstyle \max_i \endgroup E_{\max}(\path1_i') \\
      \leq {} & \begingroup \textstyle \max_i \endgroup
        \max \{ \energy{\alpha_i}, \energy{\alpha_{i + 1}} \} + M_{\bondworth} \\
      = {} & \begingroup \textstyle \max_i \endgroup
        \energy{\alpha_i} + M_{\bondworth} \\
      = {} & E_{\max}(\orig) + M_{\bondworth}
      .
  \end{align}
  So $\pathcost{\sat} \leq \pathcost{\orig} + M_{\bondworth}$.
  So $\satbarrier{\alpha_1}{\alpha_n}
  \leq \barrier{\alpha_1}{\alpha_n} + M_{\bondworth}$.
  \qedhere
\end{proof}

\noindent
Notice that we need bond strength $\bondworth \geq 1$
in Theorem~\ref{thm:satbarrier}.
If $\bondworth < 1$,
then $\satbarrier{\conf1}{\conf2}$
can be larger than $\barrier{\conf1}{\conf2}$
by an arbitrary amount.
 
Also notice that $\max \{ 0, 2 - \bondworth \}$ is tight.
To see this,
the reader may check that
$\satbarrier{\start}{\finish}
= \barrier{\start}{\finish} + \max \{ 0, 2 - \bondworth \}$
for the following example:
$\gamma = \{ \{ \mono1_1, \mono1_2 \}, \allowbreak \{ \mono1_3 \} \}$
and
$\delta = \{ \{ \mono1_1 \}, \{ \mono1_2, \mono1_3 \} \}$
where
$\mono1_1 = \{ a, b \}$,
$\mono1_2 = \{ \flip{a} \}$,
$\mono1_3 = \{ a, c \}$.

Now since $\satbarrier{\start}{\finish} \geq \barrier{\start}{\finish}$,
we have the following corollary of Theorem~\ref{thm:satbarrier},
which is the main result of this section.

\begin{cor} \label{cor:satbarrieralphatwo}
  Let bond strength $\bondworth \geq 2$
  and $\start$ and $\finish$ be saturated.
  Then
  $
    \satbarrier{\start}{\finish}
    = \barrier{\start}{\finish}
    \,.
  $
\end{cor}

%% file: modules.tex
\section{TBNs with programmable energy barriers} \seclabel{modules}

In this section we present two constructions.
Each is a family of TBNs indexed by an integer $n$.
We call certain configurations of those TBNs
\emph{initial} and \emph{triggered}
and show an energy barrier between them.
As $n$ increases,
the size of the energy barrier increases linearly.
Each also has a catalyst,
which reduces the energy barrier to 1 when added.

The first construction (\translatorcycle), 
discussed in Section~\ref{subsec:translator},
is based on a DNA strand displacement catalyst.
Progress from the initial to triggered configurations with the catalyst can be physically implemented as a strand displacement cascade.
Although this system has been previously proposed%
\cite{thachuk2015leakless,wang2018}, 
we rigorously prove an energy barrier for the first time.

The second construction (\hashgate),
discussed in Section~\ref{sec:hashgate},
does not have an evident physical implementation (e.g., as a strand displacement system),
but surpasses the \translatorcycle\ system
in two ways. First, the \hashgate\ can exhibit \emph{autocatalysis}---that is, it can be modified so that the catalyst transforms the initial polymer into a polymer that has the same exposed binding sites as the catalyst, which can itself catalyze the transformation of additional initial polymers (leading to exponential amplification).
Second, the \hashgate\ is \emph{self-stabilizing}, 
which we define to mean that
from any configuration, there is a barrier of zero to reach either an initial or triggered configuration. 
This intuitively ensures that the system cannot get stuck in an undesired local energy minimum. 

Throughout both sections,
we assume $\bondworth \geq 2$,
so that by Corollary~\ref{cor:satbarrieralphatwo},
we can determine energy barriers by considering only saturated paths.
If we weaken this assumption to $\bondworth \geq 1$,
then by Theorem~\ref{thm:satbarrier} the barrier proved is within $1$ of the barrier in the unrestricted pathway model (allowing unsaturated configurations).
We believe that for these systems an $\Omega(n)$ energy barrier exists even if $w < 1$ but sufficiently large. 
However, studying the $w < 1$ regime is left for future work (see \secref{conc}).

The constructions demonstrate that catalysts and autocatalysts with arbitrarily high energy barriers can be engineered solely by reference to the general thermodynamic driving forces of binding and formation of separate complexes, which are captured in the TBN model.

%% file: translator.tex
\subsection{\Translatorcycle}
\label{subsec:translator}

\begin{figure}[t]
  \centering
  \includegraphics[width=0.9\columnwidth]{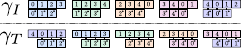}\
  \caption{The two stable configurations of a \translatorcycle\ with complex length $\strandlength = 3$ and number of complex types $\fuelcount = 5$. In place of binding site $x_i$, we write $i$ for clarity.
  }
  \figlabel{translator_sat_configs}
\end{figure}

\begin{figure}[t]
  \centering
  \opt{llncs}{\def\ratio{0.74}}
  \opt{ieee}{\def\ratio{0.79}}
  \includegraphics[width=\ratio\columnwidth]{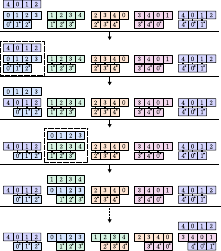}%
  \caption{A segment of the height $1$ path which is possible because an extra copy of a top monomer, $\{x_4, x_0, x_1, x_2\}$, is present to act as a catalyst.
  The dotted arrow signifies a sequence of merge/split steps. In place of binding site $x_i$, we write $i$ for clarity.}
  \figlabel{translator_ell_barrier}
\end{figure}
\noindent
\noindent
Consider the TBN illustrated in \figref{translator_sat_configs}.
There are two particular configurations that interest us,
an initial configuration $\init$ and a triggered configuration $\trig$.
The two configurations are stable.
In the presence of a catalyst monomer $\{x_4, x_0, x_1, x_2\}$ (or an extra copy of any \emph{top monomer}---any of the monomers with unstarred binding sites),
a height one pathway exists to reach $\trig$, illustrated by \figref{translator_ell_barrier}.
If the catalyst is not present, we prove there is a barrier which can be made arbitrarily large by including more and longer monomer types.
Further, this catalytic system is realizable as a DNA strand displacement cascade;
more information about this connection can be found in Appendix~\ref{app:subsec:strand_displacement}, and in~\cite{wang2018}.
In the case of many copies of each complex, since the catalyst is in fact any of the top monomers, the system may be used as an amplifier: at the end of the pathway shown in \figref{translator_ell_barrier}, another monomer with binding sites $\{x_4, x_0, x_1, x_2\}$ becomes free which can catalyze another set of complexes which are in the initial configuration.

To program a large energy barrier, we give a formal definition for generalizing the translator cycle, parameterized by \emph{complex length} $\strandlength$ and \emph{number of complex types} $\fuelcount$.
Given $\strandlength \leq \fuelcount$, a \emph{$(\strandlength$, $\fuelcount)-$translator cycle} is a TBN with monomer types
$t_i$ (\emph{top monomers}) and $b_i$ (\emph{bottom monomers}) for $i \in \mathbb{Z}_{\fuelcount}$, where
$$t_i = \{ x_i, x_{i+1 \pmod{\fuelcount}}, \ldots, x_{i+\strandlength \pmod{\fuelcount}} \},$$
$$b_i = \{x_i^*, x_{i+1 \pmod{\fuelcount}}^*, \ldots, x_{i+\strandlength-1 \pmod{\fuelcount}}^*  \}.$$
A $(\strandlength$, $\fuelcount)-$translator cycle may have any number of each monomer type as long as $(1)$ for all $i$, the number of $t_i$ is equal to the number of $b_i$ and $(2)$ there is at least one of each $t_i$ and $b_i$.
To justify constraint $(1)$ note that including an extra top monomer can catalyze the cycle so the barrier disappears, while extra bottom monomers merge the two-monomer complexes to saturate, disrupting the desired initial and triggered configurations.
Constraint $(2)$ is required for the catalytic pathway (\figref{translator_ell_barrier}).

The \emph{initial configuration} has each $b_i$ in a polymer $\{b_i, t_i\}$, and a \emph{triggered configuration} $\trig$ is any saturated configuration which contains a subset $\big\{ \{b_i, t_{i-1}\} \mid i \in \mathbb{Z}_{c} \big\}$ (at least one set of complexes in the triggered state).
The rest of this section is dedicated to proving that the barrier between $\init$ and $\trig$ depends on the complex length
$\strandlength$ and the number of complex types $\fuelcount$,
and can be made arbitrarily large: Formally, we prove that if $\strandlength^2 = \fuelcount$, then $\barrier{\init}{\trig} > \frac{z}{2+z^{-1}}$.

To motivate our choice of $\strandlength^2 = \fuelcount$, it is worth describing three available pathways that influence how the upper bound for the uncatalyzed barrier scales with $\strandlength$ and $\fuelcount$.
For the first path, we can reach a configuration with a free top monomer, which can subsequently be used as a catalyst as in \figref{translator_ell_barrier}. 
For example, merging the three polymers with exposed sites $x_4$, $x_0$, and $x_1$ with the polymer $\{\{x_4, x_0, x_1, x_2\},\{x_4^*, x_0^*, x_1^*\}\}$ allows the top monomer $\{x_4, x_0, x_1, x_2\}$ to be split.
In the general case, this requires merging $\strandlength$ polymers and then following a height $1$ path, and the path in total has height $\strandlength$.
The second path brings all complexes together while reducing the number of polymers by $\fuelcount-1$, and then splits them into the triggered complexes, resulting in a height $\fuelcount-1$ path.
These paths show that the barrier is not larger than $z$ or $c-1$.
Surprisingly, it is still not sufficient to set $\strandlength = \fuelcount-1 = n$ to attain a barrier of $n$; there is a complicated third path which
has height $\frac{2\fuelcount}{\strandlength}-1$ which is illustrated in the appendix in \figref{translator_f_barrier}.
The third path shows that $\fuelcount$ must be asymptotically larger than $\strandlength$ to achieve a non-constant barrier.
Thus we fix $z = n$ and $c = n^2$ for the remainder of the section.

Before we get into details, we give an overview of the proof that $\barrier{\init}{\trig} > \frac{n}{2+n^{-1}}$.
First, we show that we can restrict our attention to configurations where polymers must have the same amount of top monomers as bottom monomers (denoted as \emph{normal form}),
since other configurations have low polymer count.
We think of pairing top and bottom monomers in normal form polymers.
Initially, top and bottom monomers with the same indices are paired.
In the triggered configuration, the top and bottom monomers are paired with different indices, notably $t_i$ is bound to $b_{i+1}$; we say the top index is offset ``to the left'' of the bottom index by one, or has offset ``minus one''.
We will formalize this notion of offset,
and show that the sum of all offsets between pairs in the configuration, initially zero, does not change with merges and splits in paths of normal form configurations.
In the single-copy case, this contradicts any path which reaches a triggered configuration, which must have a negative total offset.

In the multi-copy case, the negative offset of the triggered complexes can be canceled by positive offset elsewhere, resulting in zero total offset, so the argument is not as simple. 
We will show that polymers providing net positive offset have a size which grows along with the offset.
The large size of these polymers then implies a barrier.

First, we provide an easy way to count how many complexes are in a configuration of a cycle.
Note that since there is at least one unstarred binding site for each starred site, the starred sites \emph{must} be bound in a saturated configuration.
So we call the starred sites \emph{limiting}.
We can use this fact to argue about the number of separate polymers in saturated configurations.
Since the starred binding sites are limiting and $\conf1$ is saturated, each bottom monomer must be bound to at least one top monomer, so we can count $S(\conf1)$ by counting the number of top monomers in separate polymers.
This leads to:
\begin{observation}\label{clm:k_tops}
Given a $(z,c)$-translator cycle with $k$ instances of top monomers, in a saturated configuration $\conf1$, if there is a polymer with $m$ top monomers, then
$S(\conf1) \leq k-m+1$.
\end{observation}

Now we restrict the configurations and paths we must consider by describing a \emph{normal form} for polymers.
\begin{definition}
Given a configuration of an $(n, n^2)$-translator cycle, a polymer is \emph{normal form} if its number of top monomers is equal to its number of bottom monomers.
A configuration is normal form if every polymer is normal form.
A path is normal form if every configuration is normal form.
\end{definition}

Normal form paths are more restricted than arbitrary paths, and will be easier to reason about. To motivate them, we show that saturated paths from $\gamma_I$ to $\gamma_T$ that are not normal form must have a large height via a large polymer in some configuration, and so low height paths (if they existed) would be normal form.

The following lemma is a technical fact used in proofs of  Lemmas~\ref{clm:translator_u<s} and~\ref{lem:perfect_matching}.
It gives properties for polymers in saturated configurations with $x$ bottom monomers with $n$ binding sites each, and $y$ top monomers with $(n+1)$ sites each, with $x > y$.
This will help us restrict to normal form polymers, which have $x = y$.

\begin{lem}\label{clm:xyn}
Assume $x, y, n \in \mathbb{N}$, $y(n+1) \geq xn$, and $x > y$. Then $y \geq n$ and $x \geq n+1$. 
\end{lem}
\begin{proof}
If $x > y$, then $x \geq y+1$, so $y(n+1) \geq xn \geq (y+1)n$.
So $yn + y \geq yn + n$.
So $y \geq n$.
Since $x > y$, $x \geq n+1$.
\qedhere
\end{proof}
The next lemma shows the saturated configurations which are not normal form must have a large (size $\Omega(n)$) polymer.
\begin{lem}\label{clm:translator_u<s}
If a saturated configuration is not normal form,
then some polymer has at least $n$ top monomers.
\end{lem}

\begin{proof}
Since the configuration is not normal form,
some polymer $\poly1$ has either fewer or more top than bottom monomers.
If $\poly1$ has fewer,
then let $\poly1' = \poly1$.
If $\poly1$ has more,
then some other polymer $\poly1'$ has fewer top monomers than bottom monomers.
Let $t$ (resp., $b$) be the number of top (resp., bottom) monomers in $\poly1'$.
The number of unstarred (resp., starred) binding sites in $\poly1'$ is $t(n+1)$ (resp., $b n$).
Recall that the starred sites are limiting.
So for each starred site in $\poly1'$,
there is at least one corresponding unstarred site,
so $b n \leq t (n + 1)$.
Since $\poly1'$ has $b > t$, Lemma~\ref{clm:xyn} gives $t \geq n$, so $P'$ has at least $n$ top monomers.
\qedhere
\end{proof}
\input{translator_proof_v2.tex}

%% file: translator_proof_v2.tex
So saturated paths with low height must consist of normal form configurations, since otherwise they would have a polymer with many top monomers
which implies a large height by Observation~\ref{clm:k_tops}.

Now we formalize the offset of a pair of compatible monomers.
(Recall that two monomers are compatible if they have complementary binding sites.)
For $k \in \mathbb{N}$ and $a,b \in \mathbb{Z}_k$,
define the sequence
\newcommand{\smod}[3]{[#1, #2]_{#3}}
\begin{equation}
  \smod{a}{b}{k} = \langle a, a + 1, \ldots, b \rangle \pmod{k}
  .
\end{equation}
For example,
$\smod{1}{3}{5} = \langle 1, 2, 3 \rangle$
and $\smod{3}{1}{5} = \langle 3, 4, 0, 1 \rangle$.
Also let $\ell_S$ be the index of element $\ell$ in sequence $S$.
Then for monomers $b_i$ and $t_j$,
we define the \emph{offset} to be $f(b_i, t_j) = j_S - i_S$,
where $S = \smod{i - n}{i + n - 1}{n^2}$.

We will define the offset of a normal form polymer in terms of compatible pairs of top and bottom monomers.
To choose the pairs, we use the notion of \emph{matchings} from graph theory.
Given a normal form polymer $P$, let $T$ be the set of top monomers and $B$ be the set of bottom monomers.
Define a bipartite graph $G_P = \{T,B,E\}$ where $\{b_i, t_j\} \in E$ if and only if $b_i$ and $t_j$ are compatible.

\begin{lem}\label{lem:perfect_matching}
If a polymer $P$ is in a normal form saturated configuration and has size $|P| < 2n+1$, then there exists a perfect matching on $G_P$.
\end{lem}

The proof of Lemma~\ref{lem:perfect_matching} is in Appendix~\ref{app:translator}.

For a perfect matching $M$ on $G_P$,
the matching offset is $f(M) = \sum_{m\in M} f(m)$.
It will turn out that for small polymers, the offset of every perfect matching is the same.
So we will use the matching offset to define the offset of a polymer.
To do so, it will be useful to define a notion of ``leftmost'' and ``rightmost'' monomers in a polymer.
For small polymers (of size less than about $n$), these are intuitively well-defined since there are not enough monomers to ``wrap-around'' the entire cycle.
We formalize this notion via a \emph{cutoff} value for a polymer:
\begin{definition}\label{def:cutoff}
Given a normal form polymer $P$, a \emph{cutoff value} $c_P \in \mathbb{Z}_{n^2}$ satisfies the following: let $C$ be $\smod{c_P}{c_P-1}{n^2}$, then there is no edge $\{b_i$, $t_j\} \in G_P$ such that $i_C \leq n$ and $j_C > n^2 - n$ or $j_C \leq n$ and $i_C > n^2 - n$ (recall that $i_C$ denotes the index of $i$ in $C$).
\end{definition}
If a cutoff value exists, it means no possible bonds in the polymer---and equivalently no possible edges in any matching---cross the cutoff.
Then the leftmost and rightmost pairs are easily defined with respect to the cutoff sequence $\smod{c_P}{c_P-1}{n^2}$.
We prove a sufficient condition for there to exist a cutoff point:

\begin{lem}\label{clm:cutoff}
For a normal form polymer $P$ with size $|P| < \frac{2n^2}{2n+1}$, there exists a cutoff value $c_P$.
\end{lem}
\begin{proof}
First, we give the reasoning behind the choice of $|P| < \frac{2n^2}{2n+1}$.
Intuitively, we want to choose a cutoff point at an index where the top monomer with that index is not compatible for any bottom monomer in the polymer.
If there are $k$ bottom monomers in a polymer, the union of the sets of compatible top monomers for those bottom polymers is at most $k(2n+1)$.
If we have $k(2n+1) < n^2$, there will exist an index for a top monomer which is not compatible to any bottom in the polymer.
This gives $k < \frac{n^2}{2n+1}$, and there are $2k$ monomers in the polymer, so we must have $|P| = 2k < \frac{2n^2}{2n+1}$.
Then let $c_P$ be the index of that incompatible top monomer.

We now show that $c_P$ is in fact a cutoff point as in Definition~\ref{def:cutoff}.
Towards contradiction, let $C$ be $\smod{c_P}{c_P-1}{n^2}$ and assume there does exist an edge $\{b_i$, $t_j\} \in G_P$ such that $i_C \leq n$ and $j_C > n^2 - n$ or $j_C \leq n$ and $i_C > n^2 - n$.
Then $t_j$ is compatible for $b_i$, but $c_P$ is in between $i$ and $j$, so then $t_{c_P}$ is a compatible top monomer for $b_i$, which contradicts our choice of $c_P$ as an index of an incompatible top monomer for any bottom monomer in the polymer.
\qedhere
\end{proof}

To use the cutoff value in future lemmas, we consider configurations and paths with polymers' size restricted to less than $\frac{2n^2}{2n+1}$.
We let $n' = \frac{2n^2}{2n+1}$ and define the following:
\begin{definition}
A polymer is $n'$-bounded if its size is less than $n'$.
A configuration is $n'$-bounded if every polymer is $n'$-bounded.
A path is $n'$-bounded if every configuration is $n'$-bounded.
\end{definition}

We can use a cutoff to show that the offset of every matching is the same.

\begin{lem}\label{clm:matchings_equal}
For an $n'$-bounded polymer $P$ in a normal form saturated configuration, for any two perfect matchings $M$ and $M'$ on $G_P$, $f(M) = f(M')$.
\end{lem}

The proof of Lemma~\ref{clm:matchings_equal} is in Appendix~\ref{app:translator}.

So we define the \emph{polymer offset} of $P$, $f(P)$, simply as the offset of any perfect matching on $G_P$.
Given a configuration $\gamma$, we can
define the \emph{configuration offset} as $\sum_{P \in \gamma} f(P)$.
We now show that under certain conditions,
merges and splits do not change the configuration offset.

\begin{lem}\label{clm:translator_config_score_constant}
In a normal form saturated $n'$-bounded path, if $\gamma \sxanymerge \delta$, then $f(\gamma) = f(\delta)$.
\end{lem}
\begin{proof}
Consider the two polymers which merge, $P_1$ and $P_2$, and call the polymer which is their union $P$.
Let $M_{P_1}$ and $M_{P_2}$ be any perfect matchings of $G_{P_1}$ and $G_{P_2}$.
Then $M_{P} = M_{P_1} \cup M_{P_2}$ is a perfect matching on $G_P$, and $f(M_{P}) = f(M_{P_1}) + f(M_{P_2})$.
Since the polymers are $n'$-bounded, we have by Lemma~\ref{clm:matchings_equal} that the polymer offsets equal the offsets of any matching, so $f(P) = f(P_1) + f(P_2)$.
Then $f(\gamma) = f(\delta)$ since their only different summands are $f(P)$ and $f(P_1), f(P_2)$.
\qedhere
\end{proof}

From the above lemma, one can prove that in the case of one copy of each monomer type, the $(n, n^2)-$translator cycle has a barrier of $\Omega(n)$ to trigger.
An informal proof follows: 
any path which is not normal form or does not have small polymers must have a large height due to a polymer with many top monomers (see Observation~\ref{clm:k_tops}). Otherwise, if we restrict paths to saturated normal form $n'$-bounded paths, Lemma~\ref{clm:translator_config_score_constant} gives that from the initial configuration offset of zero, there is no path which can change the configuration offset to $-n^2$, which is the offset of the triggered configuration. We leave out a formal statement of this proof for brevity, as we prove a more general theorem later---a barrier in the multi-copy case. 

The above argument is not sufficient in the multi-copy case because the $-n^2$ offset given by triggered polymers can be canceled out by positive $n^2$ offset in other polymers. 
We will argue that having a large positive offset is (roughly) proportional to having a large polymer or several large polymers, so the $n^2$ positive offset would require merging many complexes.
To do so, we define and prove existence of a \emph{sorted} matching on polymers.
Intuitively, a sorted matching is a matching which has no crossing edges when the indices are sorted.

\begin{lem}\label{clm:sorted}
Given an $n'$-bounded polymer $P$ with cutoff $c_P$, there exists a \emph{sorted} matching $M$ on $G_P$ which satisfies that there does not exist $\{b_{i_1}, t_{j_2}\},\{b_{i_2}, t_{j_1}\} \in M$ with $i_1 \leq i_2$ and $j_1 \leq j_2$ with respect to the ordering given by the cutoff value sequence, $\smod{c_P}{c_P-1}{n^2}$.
\end{lem}
The proof of Lemma~\ref{clm:sorted} is in Appendix~\ref{app:translator}.

Using the sorted matching, we show that the maximum offset (if it is positive) of any one pair in a polymer is proportional to the size of the polymer.
First we relate exposed sites to the size of a polymer, then relate the exposed sites to the maximum offset.

\begin{lem}\label{clm:exposed}
If a polymer $P$ in a saturated normal form configuration has $k$ exposed sites, it has size $2k$.
\end{lem}
\begin{proof}
Assume $P$ is of size $2s$ for some $s$.
Then $P$ has $sn$ starred domains which must be bound in a saturated configuration.
$P$ has $s(n+1)$ unstarred domains.
So $P$ has exactly $s$ exposed sites for any $s$.
Then to have $k$ exposed sites, it must have size $2k$.
\qedhere
\end{proof}

\begin{lem}\label{clm:max_is_size}
Given a normal form $n'$-bounded polymer $P$, consider the sorted matching $M$ of $G_P$.
Let $m$ be the value of the maximum offset of any pair in $M$, then $|P| \geq 2(m+1)$.
\end{lem}
\begin{proof}
We will show that $P$ has at least $m+1$ exposed sites, and thus by Lemma~\ref{clm:exposed} is of size at least $2(m+1)$.
Consider the pairs ordered by smallest bottom index to largest with respect to the cutoff $c_P$ given by Lemma~\ref{clm:cutoff}.
Imagine constructing $P$ by adding one pair at a time in order.
We will show that when adding a pair, the number of exposed sites cannot decrease due to the order, and when we add the pair with offset $m$, the constructed polymer has $m+1$ exposed sites.

First, consider adding a pair $p$ with nonnegative offset $f(p)$ to a polymer with $k$ exposed unstarred sites.
Note that the polymer containing only the two monomers in the pair $p$ has $f(p)+1$ exposed unstarred sites, and $f(p)$ exposed starred sites.
By the ordering, we know that the $f(p)+1$ exposed unstarred sites cannot be bound by any bottom monomers in the polymer constructed thus far.
Further, at most $f(p)$ of the $k$ exposed unstarred sites on the polymer prior to adding $p$ can be bound after adding $p$, since only $f(p)$ exposed starred sites are added.
Thus the net change in exposed unstarred sites is plus one.
Further, the number of exposed sites on the constructed polymer after adding $p$ is at least $f(p)+1$.

Next, consider adding a pair $p$ with negative offset to a polymer with $k$ exposed unstarred sites.
Note that due to the ordering, the bottom monomer in the pair has no domains in common with the unstarred exposed sites of the polymer constructed thus far.
So the number of exposed sites does not decrease.

In both cases, the number of exposed unstarred sites cannot decrease by adding polymers.
Consider the point in this construction where we have just added the polymer with positive offset $m$.
The constructed polymer thus far has the $m+1$ exposed unstarred sites given by the offset $m$,
and the number of exposed sites cannot be reduced by adding the remaining polymers, so the final polymer $P$ must have at least $m+1$ exposed sites.
So by Lemma~\ref{clm:exposed}, $|P| \geq 2(m+1)$.
\qedhere
\end{proof}

We prove two more lemmas before the proof of the barrier of the translator cycle. 
The first key lemma is that triggered polymers' negative offset must be canceled out by polymers with positive offset, but since positive offset results in large polymers (or many slightly larger polymers), such a configuration implies a large height for the path which contains it.

\begin{lem}\label{clm:positive_offset_large_polymer}
Given a normal form saturated $n'$-bounded configuration $\gamma$, if there exists a subset of polymers $\mathcal{P} = \{P_1,\ldots,P_k\}$ such that $\sum_{P_i \in \mathcal{P}} f(P_i) \geq n^2$, then $S(\gamma_I) - S(\gamma) > 2n+1$.
\end{lem}
\begin{proof}
Consider $\gamma$ and for each polymer $P_i$ in $\gamma$, fix any sorted matching $M_i$ on $G_{P_i}$ given by Lemma~\ref{clm:sorted} and denote the set of $M_i$ by $\mathcal{M}$.
Then for each $P_i \in \mathcal{P}$,
$$f(P_i) \leq \frac{|P_i|}{2} \max_{p \in M_i}f(p),$$ since there are $\frac{|P_i|}{2}$ pairs each with offset at most the max over the offsets.
Since $\gamma$ is $n'$-bounded, $|P_i| < n'$, so $$f(P_i) < \frac{n'}{2} \max_{p \in M_i}f(p).$$
So  we have the following:
\begin{align}
\sum_{M_i \in \mathcal{M}} \frac{n'}{2} \max_{p \in M_i}f(p) &=  \frac{n'}{2} \sum_{M_i \in \mathcal{M}} \max_{p \in M_i}f(p) \\ &> \sum_{P_i \in \mathcal{P}} f(P_i) \\ &\geq n^2, 
\end{align}
and further \begin{align}
\sum_{M_i \in \mathcal{M}} \max_{p \in M_i}f(p) > \frac{2n^2}{n'}.
\eqlabel{quadmax}
\end{align}

Now we show that $$S(\gamma_I) - S(\gamma) \geq \sum_{M_i \in \mathcal{M}} \max_{p \in M_i}f(p).$$
Consider $\gamma'$, the (unsaturated) configuration which is given by taking the polymers in $\gamma$ and splitting them into pairs of top and bottom monomers based on the matchings $M_i$.
Each bottom monomer is in a polymer with exactly one top monomer, so $S(\gamma_I) = S(\gamma')$.
For each $P_i$ with sorted matching $M_i$ in $\cal{P}$, consider the pair $p$ satisfying $\max_{p \in M_i}f(p)$.
We know that in $\gamma$, each polymer $P_i$ with $p \in M_i$ must have size at least $2(f(p)+1)$ by Lemma~\ref{clm:max_is_size}.
So $P_i$ must be a polymer containing at least $f(p)$ other pairs.
So $$S(\gamma') - S(\gamma) \geq \sum_{M_i \in \mathcal{M}} \max_{p \in M_i}f(p).$$
Since $S(\gamma_I) = S(\gamma')$, \eqref{quadmax} gives us $S(\gamma_I) - S(\gamma) > \frac{2n^2}{n'} = 2n+1$.
\qedhere
\end{proof}

\begin{theorem}
Given an $(n, n^2)$-translator cycle, $b(\gamma_I,\gamma_T) \geq \frac{n}{2+n^{-1}}$.
\end{theorem}
\begin{proof}
We split the possible paths from $\gamma_I$ to $\gamma_T$ into three cases.
Case 1: if a saturated path is normal form and $n'$-bounded, we have by Lemma~\ref{clm:translator_config_score_constant} that for any $\gamma_i$ in the path, $f(\gamma_i) = f(\gamma_I) = 0$.
Consider the final configuration on the path, the triggered configuration $\gamma_T$.
By definition, in $\gamma_T$ there exists at least one pair of each $b_i$ bound to $t_{i-1\pmod n^2}$.
Each of these pairs has offset minus one, contributing minus $n^2$ to the offset.
However, we know that the configuration offset must be zero.
So there must exist a subset of polymers $\mathcal{P} = \{P_1,\ldots,P_k\}$ such that $\sum_{i=1}^k f(P_i) \geq n^2$, so by Lemma~\ref{clm:positive_offset_large_polymer}, $S(\gamma_I) - S(\gamma_T) > 2n+1$.

Case 2:
if a path is not normal form, then by Lemma~\ref{clm:translator_u<s}, there exists a configuration $\gamma$ with a polymer with $n$ top monomers.
Let $k$ be the total number of top monomers in the cycle and note that $S(\gamma_I) = k$.
Then by Observation~\ref{clm:k_tops}, $S(\gamma) \leq k-n+1$.
Then $S(\gamma_I) - S(\gamma) \geq n-1$.

Case 3: if a path is normal form but is not $n'$-bounded, then by definition there exists a polymer of size at least $n'$ with an equal number of top and bottom monomers; i.e., a polymer with at least $\frac{n'}{2}$ top monomers.
Let $k$ be the total number of top monomers in the cycle and note that $S(\gamma_I) = k$.
Then by Observation~\ref{clm:k_tops}, $S(\gamma) \leq k-\frac{n'}{2}+1$.
Note that $S(\gamma_I) = k$.
Then $S(\gamma_I) - S(\gamma) \geq \frac{n'}{2}-1 = \frac{n^2}{2n+1} - 1$.

By Corollary~\ref{cor:satbarrieralphatwo}, we restrict analysis to saturated paths.
Then $H(\gamma_I)=H(\gamma)$, and so $E(\gamma)-E(\gamma_I)$ = $-\left(S(\gamma_I)-S(\gamma)\right)$.
Among the three cases, the smallest lower bound on the height is $\frac{n^2}{2n+1}$, so the barrier is at least $\frac{n^2}{2n+1} = \frac{n}{2+n^{-1}}$.
\qedhere
\end{proof}

%% file: hashgate.tex
\input{notation_hashgate}

\subsection{\Hashgate}
\label{sec:hashgate}

{\input{fighashgate}}
\noindent
Consider the TBN illustrated in \figref{hashgate}.
We focus on two polymer types $G_H$ and $G_V$ depicted in the figure,
and show that there is a barrier $n \in \mathbb{N}$ to convert $G_H$ to $G_V$ and vice versa.
We parameterize the construction by $n$ as follows.
 Define the following monomer types:
 ``horizontal'' $H_i \defeq \{x_{ij}\}_{j=1}^n$ for $i \in \{1,\dots,n\}$,
``vertical'' $V_j \defeq \{x_{ij}\}_{i=1}^n$ for $j \in \{1,\dots,n\}$, and
 ``gate'' $G \defeq \{x_{ij}^*\}_{i,j=1}^n$.  
In the notation of chemical reaction networks, 
the net reaction
\[
G_H \rightleftharpoons G_V
\]
can occur in the presence of sufficiently many free $H_i$'s and $V_j$'s,
but an energy barrier of $n$ must be surmounted in order for this conversion to happen.  In Section \ref{hashgate:catsection} we will show how this energy barrier can be reduced to $1$ in the presence of a \emph{catalyst} monomer,
corresponding to the chemical notion of a catalyst reducing the activation energy required for a reaction to occur.

Note that throughout this section, the configurations considered are saturated, so that for two configurations $\conf{1}$ and $\conf{2}$, 
we have $H(\conf{1})=H(\conf{2})$, and so $E(\conf{2})-E(\conf{1})$ = $-\left(S(\conf{2})-S(\conf{1})\right)$.  That is, the energy difference between two configurations is the opposite of the difference in their polymer counts.

We fix a network $\networkgg$ that contains any number of any of these monomer types, so long as there are enough of other monomers to completely bind all the $G$ monomers (i.e., in saturated configurations there are no exposed starred sites).
We define \emph{base} configurations of the network to be those configurations that contain polymers of type $G_H$ or $G_V$, with all other monomers in separate polymers by themselves.
In Theorem \ref{thm:copy_tolerant_grid_catalyst_stability} we show that these base configurations are stable (take $m=0$).

The following lemma establishes a necessary condition in any saturated configuration: that any $G$ must be in a polymer with either all of the horizontal monomers or all of the vertical monomers.

\begin{lem} \label{lemma:complete-g}
  In a saturated configuration of $\networkgg$,
  a polymer containing $G$
  also contains $\allh$ or $\allv$ as a subset.
\end{lem}

\begin{proof}
  Suppose a polymer contains $G$
  but neither $H_i$ nor $V_j$ for some $i$ and $j$.
  Then site $x_{ij}^*$ on $G$ is exposed, and so by definition of $\networkgg$, the configuration is not saturated.
  \qedhere
\end{proof}

The following theorem then establishes that any saturated configuration in $\networkgg$ is \emph{self-stabilizing}, that is, it can reach a stable (base) configuration via a path with barrier $0$ (i.e. using all splits).

{
\newcommand{\sat}{\conf1}%

\begin{theorem} \label{self_stabilize_to_base}
  For any saturated $\sat$ of $\networkgg$,
  some base $\base$ has $\sat \cleansplit \base$.
\end{theorem}
\begin{proof}
  \newcommand{\lump}{\poly1}%
  Consider a saturated configuration $\sat$.
  Suppose $\sat$ has a non-base polymer $\lump$.
  If $\lump$ contains no $G$,
  then we can split into polymers of type $H$, $V$.
  Otherwise, $\lump$ contains $G$,
  and by Lemma~\ref{lemma:complete-g}
  we can split off a $G_V$ or $G_H$ polymer.
  The theorem holds by repeating this process.
  \qedhere
\end{proof}
}

Note that since the base configurations are stable
(this follows from Theorem~\ref{thm:copy_tolerant_grid_catalyst_stability} with $m=0$), 
and by Theorem \ref{self_stabilize_to_base} any other saturated configuration can reach a base configuration using only splits, the base configurations are also the only stable configurations in this network.

We now prove the desired energy barrier between different base configurations.

\begin{theorem}\label{thm: grid_barrier}
  The barrier between different base configurations of $\networkgg$ is $n$.
\end{theorem}

\begin{proof}
  \newcommand{\sat}{\path1}%
  \newcommand{\start}{\conf1}%
  \newcommand{\finish}{\conf2}%
  \newcommand{\firstbad}{\conf4}%
  \newcommand{\peak}{\conf3}%
  
  Consider a saturated path $\sat$
  from a base configuration $\start$ to another, $\finish$.
  Notice that $\finish \not\anysplit \start$.
  So $\finish \not\cleansplit \start$.
  So some first $\firstbad$ along $\sat$
  has $\firstbad \not\cleansplit \start$.
  But by Theorem~\ref{self_stabilize_to_base},
  some other base $\base \neq \start$ does have $\firstbad \cleansplit \base$.
  
  Now take $\peak$ just before $\firstbad$ along $\sat$.
  Then $\peak \cleansplit \start$ by definition of $\firstbad$.
  Since $\alpha$ and $\beta$ are adjacent on $\sat$, either $\peak \sxcleansplit \firstbad$
  or $\peak \sxcleanmerge \firstbad$.
  The latter contradicts $\firstbad \not\cleansplit \start$.
  So $\peak \sxcleansplit \firstbad$,
  implying $\peak \cleansplit \base$.
  
  \newcommand{\boundh}[1]{f(#1)}%
  Let $\boundh{\start}$ count the $H$ monomers
  with a $G$ in $\start$.
  Since $\start$ and $\base$ are different bases,
  \emph{wlog}, $\boundh{\base} \geq n + \boundh{\start}$.
  Consider a path of $k$ merges
  corresponding to $\start \cleanmerge \peak$.
  It can increase $\boundh{\cdot}$ by at most $k$.
  So $\boundh{\peak} \leq \boundh{\start} + k$.
  A path of splits does not increase $\boundh{\cdot}$,
  so $\peak \cleansplit \base$ implies $\boundh{\peak} \geq \boundh{\base}$.
  So we get
  \begin{align}
    k
    \geq{}& \boundh{\peak} - \boundh{\start}
    \\
    \geq{}& \boundh{\base} - \boundh{\start}
    \geq n
    .
  \end{align}
  So $\polycount{\start} - \polycount{\peak} = k \geq n$.
  \qedhere
\end{proof}

\setcounter{secnumdepth}{3}

\subsubsection{Catalysis}\label{hashgate:catsection}

{\input{fighashgatecats}}

{\input{fighashmechanismfull}}

\noindent
The kinetic barrier shown for the \hashgate\ can be disrupted by the presence of new monomer types.
In fact, the model admits a \emph{catalyst} monomer $C$ that lowers the energy barrier from $n$ to 1, 
i.e., in the presence of one or more $C$, a $G_H$ can be converted into a $G_V$, and vice versa, with a sequence of merge-split pairs.
In the notation of chemical reaction networks, this binding network implements the net reaction
\[
G_H + C \leftrightharpoons G_V + C
\]
with energy barrier 1,
while maintaining a large energy barrier for the reaction $G_H \leftrightharpoons G_V$.

For the \hashgate\ of size $n \times n$, we define a catalyst:
$
C \defeq \{x_{ij} \mid 1 \leq j \leq i \leq n\}$
illustrated in Fig.~\ref{fig:hashgatecats} (left).
$C$ is a monomer consisting of the ``lower triangle'' of the unstarred sites.
The mechanism by which $C$ can transform $G_H$ to $G_V$ with merge-split pairs is by an alternating processes of merges and splits shown in \figref{hashmechanismfull}.
Intuitively,
in each step of the catalyzed reaction $G_H + C \to G_V + C$, $G$ switches its association with $H_i$ (left) to its counterpart on $V_j$ (right) by merging the evolving polymer (center) with $V_j$ and then splitting off $H_i$.  

Consider a network
$\{G, C\} \cup \allh \cup \allv$
which includes one instance of every monomer type that has been introduced, as well as the catalyst.  
As before, we shall be interested in net transitions between $G_H$ and $G_V$, and so for
this network
we define the following configurations: $\confhc \defeq \{G_H, C\} \cup \allv$ and 
 $\confvc \defeq \{G_V, C\} \cup \allh$.

Theorem~\ref{hash_catalyzed} states that transitions in the single-copy case, 
having arbitrarily large energy barriers according to Theorem~\ref{thm: grid_barrier},
in the presence of $C$ have their barrier reduced to one.

\begin{theorem}\label{hash_catalyzed}
$\barrier{\confhc}{\confvc} = \barrier{\confvc}{\confhc} = 1$.
\end{theorem}

\begin{proof}
Consider the following saturated merge-split pathway that begins with configuration $\confhc$ and ends with $\confvc$ (illustrated in Figure \ref{fig:hashmechanismfull}).

\begin{itemize}
    \item Merge $G_H$ with $C$.
    \item Split $H_n$ from the resulting polymer $P_1$.
    \item For $1 \leq i \leq n-1$, iteratively merge $V_{i+1}$ with $P_{i}$, then split $H_{i}$ from $P_{i}$ to form $P_{i+1}$.
    \item Merge $V_1$ with $P_n$ and split off $C$ to form $G_V$.
\end{itemize}

This path maintains saturation while never decreasing the polymer count by more than one, and so by Corollary~\ref{cor:satbarrieralphatwo} we have that $\barrier{\confhc}{\confvc} \leq 1$.  As it is not possible to reach $\confvc$ from $\confhc$ in a saturated merge-split path that uses only splits, it will not be possible to have a zero barrier, giving $\barrier{\confhc}{\confvc} = 1$.

This merge-split path can be be executed in the reverse fashion to show that $\barrier{\confvc}{\confhc} = 1$.
\qedhere
\end{proof}

To generalize the result to the multi-copy setting, we first observe
that the height $1$ pathway guaranteed in Theorem~\ref{hash_catalyzed} still exists.
What remains is to show that the base configurations, plus zero or more separate catalyst monomers, are stable.

In the arguments that follow, it will be useful to define the set $\offdiag = \{x_{i,i+1}\}_{i=1}^{n-1} \cup \{x_{n,1}\}$, which consists of the domains from the ``shifted diagonal''.
For $m \in \mathbb{N}$, and let $\networkgg^m = \networkgg \cup \{m \cdot C\}$
denote the TBN $\networkgg$ with $m$ additional catalyst monomers.

The next lemma states that each $G$ in a polymer in a saturated configuration must be accompanied by $n$ additional monomers,
thus giving a lower bound on the size of any polymer as a function of the number of $G$'s contained within it.

\begin{lem}\label{lem:off-diag counts}
    If there are $k$ $G$'s in a polymer $P$ in a saturated configuration $\conf{2}$ of $\networkgg^m$, then $\abs{P} \geq k(n+1)$.
\end{lem}

\begin{proof}
    Consider the domains from $\offdiag$.
    As the starred versions of these domains are present on each $G$, to maintain saturation in $\conf{2}$, it must be the case that each $G$ in $\conf{2}$ is joined in $P$ with a set of monomers that include these domains; 
    however, no monomer in $\{C,H_1,\ldots,H_n,V_1,\ldots,V_n\}$
    has more than one of these domains.  
    Thus, to be saturated, if there exist $k$ instances of $G$ in $P$, 
    there must be at least $kn$ additional monomers in $P$ to bind the above chosen domains. 
    \qedhere
\end{proof}

The following lemma uses Lemma~\ref{lem:off-diag counts} to show that no saturated configuration has more polymers than are contained in the base configurations,
even when some catalyst monomers are present.

\begin{lem}\label{lem:grid_gate_max_poly_count}
    If there are $k$ $G$'s in $\networkgg^m$, then any  saturated configuration $\conf{2}$ has $\polycount{\conf{2}} \leq \abs{\networkgg^m} - kn$.
\end{lem}
\begin{proof}
    Consider the polymers $P_1, \dots, P_j$ in $\conf{2}$
    containing all $k$ copies of $G$, where $P_i$ has $k_i$ copies.
    Then by Lemma~\ref{lem:off-diag counts},
    $\sum \abs{P_i} \geq \sum k_i (n+1) = k (n+1)$.
    So $\polycount{\conf{2}}
    \leq j + \abs{\networkgg^m} - k(n+1)
    \leq k + \abs{\networkgg^m} - k(n+1)
    = \abs{\networkgg^m} - kn$.
    \qedhere
\end{proof}

The next Theorem establishes that the base configurations, which were stable in the original network $\networkgg$, are still stable even if any number of catalysts should also be present in the network.

\begin{theorem}\label{thm:copy_tolerant_grid_catalyst_stability}
Let $m \in \mathbb{N}$ and let
$\conf{1}$ be a base configuration of $\networkgg$.
Then $\conf{1}_m \defeq \conf{1} \cup \left(m\cdot\{C\}\right)$ is stable.
\end{theorem}

\begin{proof}
The base configurations satisfy Lemma \ref{lem:grid_gate_max_poly_count} with equality, thus they are stable.
\qedhere
\end{proof}

These results show that
the catalyzed network is copy tolerant; that is, it behaves in the expected way even should the amounts of the constituent monomers (and catalyst) vary.

\subsubsection{Autocatalysis}\label{app:hashgate_autocatalyst}

The \hashgate\ can also be modified to act in an \emph{autocatalytic} manner.  
By modifying the vertical monomers it is possible for $G_V$ to have a set of exposed monomers acting as a ``catalyzing region'', 
which has the same structure and function as $C$ (see Figure \ref{fig:hashgatecats}, middle and right).  

To obtain an autocatalytic system, we modify the vertical monomers of the network to include additional sites that, when combined with $G$, form a catalyzing region that can act in the same manner as the catalyst $C$.  
See Figure \ref{fig:hashgatecats} for an illustration.

Formally, we define the modified vertical monomers (see Fig.~) as:
\[
\widetilde V_j \defeq 
\{x_{ij}\}_{i=1}^n \cup \{x_{ij}\}_{i=j}^n
\]

We define $G_{\widetilde V} \defeq \{G\} \cup \alltildev$ to be the version of $G_V$ that uses the modified monomers.
This polymer is the so-called auto-catalyst.

We now consider the network $\networkggauto \defeq \{2 \cdot G\} \cup \allh \cup \twoalltildev$ which includes enough monomers to create the autocatalyst as well as retain enough monomers to analyze transitions between $G_H$ and $G_{\widetilde V}$.  The two configurations that we will be most interested in are:\\
$\tildeconfh \defeq \{G_H, G_{\widetilde V}\} \cup \alltildev$\\
$\tildeconfv \defeq \{2 \cdot G_{\widetilde V}\} \cup \allh$

The following theorem establishes that the autocatalyzed configurations $\tildeconfh$ and $\tildeconfv$ are stable, with a low energy barrier between them in the presence of the autocatalyst.

\begin{theorem}
    $\tildeconfh$ and $\tildeconfv$ are stable.
\end{theorem}
\begin{proof}
    Note that there are $3n+2$ total monomers in $\networkggauto$,
    and that $S(\tildeconfh) = S(\tildeconfv) = n+2$.
    To show that these are stable, it suffices to show any other saturated configuration $\conf{2}$ obeys $S(\conf{2}) \leq n+2$.
    Consider the set of domains
    $\offdiag = \{x_{i,i+1}\}_{i=1}^{n-1} \cup \{x_{n,1}\}$,
    as in Lemma~\ref{lem:off-diag counts}.
    Each monomer has at most one of each \emph{type} of these domains,
    and except for $x_{n,1}$,
    has exactly one instance of each.
    The exception is $\widetilde V_1$,
    which has two instances of $x_{n,1}$.
    
    Let $\conf{2}$ be any saturated configuration of $\networkggauto$.
    We consider two cases: 
    1) that the $G$'s are in separate polymers,
    and
    2) that the $G$'s are in the same polymer.
    In the case that the $G$'s are in separate polymers, 
    as each $G$ contains the starred versions of the $n$ domain types in $\offdiag$,
    each polymer $P$ with a single $G$
    must have $n$ additional monomers
    to bind each of the starred versions
    of these domains.
    Note that in this case a single $\widetilde V_1$ cannot bind both instances of $x_{n,1}^*$, for this would result in both $G$'s being on the same polymer.
    This leaves at most $(3n+2)-2(n+1) = n$ remaining monomers.
    Thus $S(\conf{2}) \leq n + 2$.

    Now consider the case that the $G$ are present in the same polymer.
    Let $P$ be a polymer in $\conf{2}$ that contains two $G$'s.
    Then in $P$, the two instances of site $x_{n,1}^*$ can be bound to one instance of $\widetilde V_1$. (Note that
    if there were two $\widetilde V_1$ in this polymer, 
    then one copy could be split into its own polymer while still maintaining saturation,
    increasing $S$ and putting us in the first case).
    The remaining sites in $\offdiag$ of these $G$'s must be bound to monomers containing $2 \cdot \{x_{i,i+1}\}_{i=1}^{n-1}$.
    No single monomer contains more than one site from this set, so this requirement must be satisfied by the inclusion of $2(n-1)$ additional non-$G$ monomers in $P$.  
    Then there are at least $2n+1$ monomers in $P$,
    leaving at most $(3n+2)-(2n+1) = n+1$ additional monomers.
    Thus $S(\conf{2}) \leq 1 + (n+1) = n+2$.
    
    Since $\conf{2}$ was arbitrary, this shows $\tildeconfh$ and $\tildeconfv$
    have maximal $S$, thus are stable.
    \qedhere
\end{proof}

In particular, the desired feature of this network is that $G_{\widetilde V}$ acts as an autocatalyst.

\begin{theorem}\label{hash_autocatalyzed_barr}
$\barrier{\tildeconfh}{\tildeconfv} = 1$.
\end{theorem}

\begin{proof}
  Figure \ref{fig:hashgatecats} shows that the exposed sites of $G_{\widetilde V}$ are exactly the sites of $C$.  The proof then follows as in Theorem \ref{hash_catalyzed}.
  \qedhere
\end{proof}

Like the catalyzed network, the results for the autocatalyzed network can be extended to the multi-copy case.
The proof follows in straightforward fashion by an inductive argument.

%% file: notation_hashgate.tex
\newcommand{\num}[2]{\#_{#1}(#2)}
\newcommand{\numfun}[1]{\#_{#1}}
\newcommand{\confv}{\conf{1}_V}
\newcommand{\confh}{\conf{1}_H}
\newcommand{\confhc}{\conf{1}_H^C}
\newcommand{\confvc}{\conf{1}_V^C}
\newcommand{\tbnwithc}{\mathcal{T}^C}
\newcommand{\tildeconfv}{\widetilde{\conf{1}}_V}
\newcommand{\tildeconfh}{\widetilde{\conf{1}}_H}
\newcommand{\allh}{\{H_i\}_{i=1}^n}
\newcommand{\allv}{\{V_j\}_{j=1}^n}
\newcommand{\alltildev}{\{\widetilde{V}_j\}_{j=1}^n}
\newcommand{\twoalltildev}{\{2\cdot\widetilde{V}_j\}_{j=1}^n}
\newcommand{\diagsites}{\{x_{ii}\}_{i=1}^n}
\newcommand{\startconf}{\conf{1}}
\newcommand{\offdiag}{D_{\uparrow}}
\newcommand{\hashbreak}{\\& \qquad\qquad\qquad}

\newcommand{\networkgg}{\mathcal{T}_{gg}}
\newcommand{\networkggauto}{\mathcal{\widetilde T}_{gg}}

\newcommand{\base}{\pi}%

\newcommand{\retheorem}[2]{\noindent\textbf{Theorem~{#1}.}\emph{#2}
}

\newcommand{\recor}[2]{\noindent\textbf{Corollary~{#1}.}\emph{#2}
}

%% file: fighashgate.tex
\begin{figure}[t]
  \centering
  \opt{llncs}{\def\ratio{0.7}}
  \opt{ieee}{\def\ratio{1}}
  \includegraphics[width=\ratio\columnwidth]{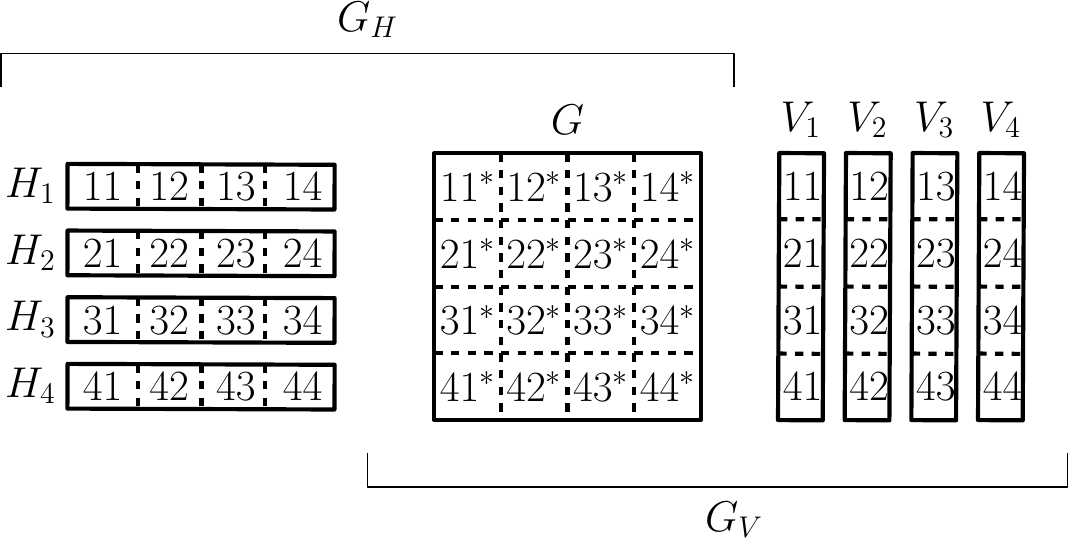}
  \caption{
    The monomer types in the \hashgate\ TBN for the case $n=4$.
    In the figure, any two digit number $ij$ represents domain $x_{ij}$, e.g. $x^*_{23}$ is represented as $23^*$.
  }
  \figlabel{hashgate}
\end{figure}

%% file: fighashgatecats.tex
\begin{figure}[t]
  \centering
  \opt{llncs}{\def\ratio{0.7}}
  \opt{ieee}{\def\ratio{1}}
  \includegraphics[width=\ratio\columnwidth]{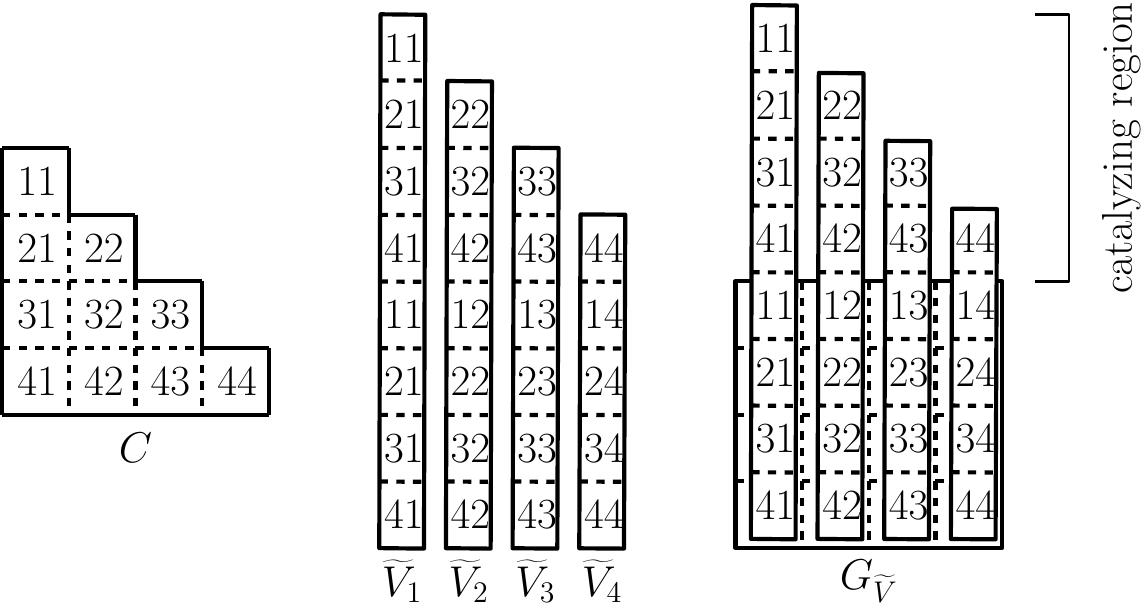}
  \caption{
    Catalysts and autocatalysts in the \hashgate\ TBN for the case $n = 4$.
    \underline{left}:
    $C$ is a single monomer that acts as a catalyst to convert between $G_H$ and $G_V$.
    \underline{middle}:
    Modified vertical monomers $\alltildev$ with extra sites.
    \underline{right}:
    After $C$ converts $G_H$ to $G_{\tilde V}$ with modified vertical monomers,
    $G_{\tilde V}$ has the same excess sites as $C$ and acts as a catalyst itself (i.e. is ``active'' as a catalyst).
  }
  \figlabel{hashgatecats}
\end{figure}

%% file: fighashmechanismfull.tex
\begin{figure}
  \centering
  \opt{llncs}{\def\ratio{0.8}}
  \opt{ieee}{\def\ratio{1}}
  \includegraphics[width=\ratio\columnwidth]{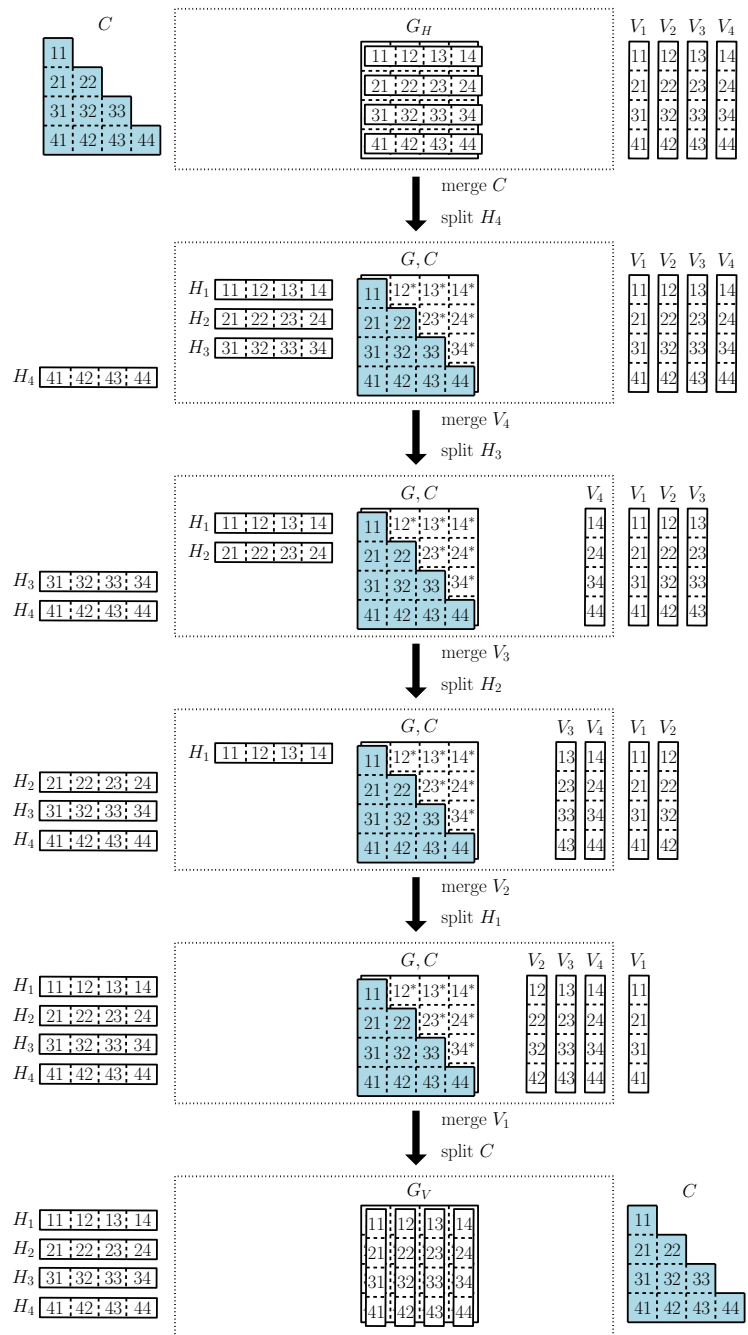}
  \caption{
    Full pathway for reaction $G_H + C \to G_V + C$.
    In each stage, exactly one merge and one split occurs, and the center polymer remains saturated.
  }
  \figlabel{hashmechanismfull}
\end{figure}

%% file: simulate.tex
\setcounter{secnumdepth}{2}

\section{Modeling bonds} \seclabel{bonds}

\begingroup

\newcommand{\start}{\conf1}%
\newcommand{\finish}{\conf2}%

The model, in \secref{model},
represents bonds implicitly.
For example, as \figref{tbnconf} shows, a single configuration can correspond to multiple ways of pairing up binding sites.
This makes it easier to manipulate and reason about configurations.

But does this simplification of configurations
affect the height of kinetic barriers?
One might imagine that by having to manipulate individual bonds,
one would need to overcome larger energy barriers than in our original model where all possible bonds are ``automatically made''.
Since bonds do change on an individual basis in a chemical system,
this would mean that a barrier that exists
in the bond-oblivious model is misleading.
However, we show that the bond-aware model is essentially equivalent.

We now define the bond-aware model analogously to the definitions of \secref{model}.

\subsection{Model}

A \defterm{configuration} $\conf1$ of a TBN
is a matching among its complementary sites
along with a partition of the components of that matching.
A \defterm{polymer} of $\conf1$
is a part of this partition.
A \defterm{bond}
is an edge in the matching.
A configuration is \defterm{saturated}
if the matching is maximal.

For a configuration $\conf1$,
denote by $\bondcount{\conf1}$%
\markdef{$\bondcount{\placeholder}$}
the total number of bonds.
Let $\polycount{\conf1}$
and $\energy{\conf1}$
be as before.

A \defterm{make} adds a missing bond in a polymer.
A \defterm{break} removes an existing bond in a polymer.
A three-way \defterm{swap} is changing one endpoint of a bond to another.
A four-way swap is swapping the endpoint of one bond with the endpoint of another.
A \defterm{path} is
a sequence of configurations
where each gets to the next
by a merge, split, make, break, or swap. 

$\pathcost{\cdot}$ is as before.
Let \defterm{$\bondbarrier{\start}{\finish}$}
be the barrier in the bond-aware model.
Let \defterm{$\satbondbarrier{\start}{\finish}$}
be over only paths with no break (and so no make).

\subsection{Equivalence}

We can view the original coarse kinetic model in \secref{model}
as a simplification of the more detailed model.
To move between the two perspectives,
we introduce some notation.
For a polymer $\poly1$ of the bond-aware model,
let $\simp{\poly1}$ be the collection of monomers in $\poly1$
(which is the corresponding polymer in the bond-oblivious model).
For a configuration $\conf1$,
let $\simp{\conf1}$ be the collection of $\simp{\poly1}$
for each polymer $\poly1$ in $\conf1$
(which is the corresponding configuration in the bond-oblivious model).

\begin{lem} \label{lem:simpenergy}
  \newcommand{\fine}{\conf3}%
  $\energy{\simp{\fine}} \leq \energy{\fine}$.
\end{lem}

\begin{proof}
  \newcommand{\fine}{\conf3}%
  \newcommand{\simpf}{\simp{\fine}}%
  $\polycount{\simpf} = \polycount{\fine}$
  and $\bondcount{\simpf} \geq \bondcount{\fine}$.
  \qedhere
\end{proof}

\noindent
The bond-aware model allows as much as the polymer model and more,
so intuitively,
a barrier in the bond-aware model should be no higher.

\begin{theorem}
  If $\energy{\start} = \energy{\simp{\start}}$,
  then $\bondbarrier{\start}{\finish}
  \leq \barrier{\simp{\start}}{\simp{\finish}}$.
\end{theorem}

\begin{proof}
  \newcommand{\pfine}{\path1}%
  \newcommand{\pcoarse}{\simp{\pfine}}%
  \newcommand{\cfine}[1]{\start_{#1}}%
  \newcommand{\ccoarse}[1]{\simp{\cfine{#1}}}%
  
  Consider a path $\pfine = \cfine{1}, \ldots, \cfine{n}$.
  Let $\pcoarse = \ccoarse{1}, \ldots, \ccoarse{n}$,
  and let
  $\ccoarse{i}$ have highest energy along $\pcoarse$.
  By Lemma~\ref{lem:simpenergy},
  $\energy{\ccoarse{i}} \leq \energy{\cfine{i}}$.
  So if $\energy{\ccoarse{1}} = \energy{\cfine{1}}$,
  then $\pathcost{\pcoarse} \leq \pathcost{\pfine}$.
  \qedhere
\end{proof}

\noindent
If $\start$ is saturated,
then $\simp{\start}$ is saturated
and $\energy{\start} = \energy{\simp{\start}}$.
So this proof also proves the inequality for saturated paths
in the two models.

\begin{lem}
  If $\start$ and $\finish$ are saturated,
  then
  $\satbondbarrier{\start}{\finish}
  \leq \satbarrier{\simp{\start}}{\simp{\finish}}$.
\end{lem}

\noindent
What may be more surprising is that we can also establish a reversed inequality.

\begin{theorem}
  $\bondbarrier{\start}{\finish} + 1
  \geq \barrier{\simp{\start}}{\simp{\finish}}$.
\end{theorem}

\begin{proof}
  \newcommand{\pcoarse}{\path1}%
  \newcommand{\pfine}{\pcoarse'}%
  Consider a path $\pcoarse$ from $\simp{\start}$ to $\simp{\finish}$.
  Form a path $\pfine$ from $\start$ to $\finish$
  by adding makes, breaks, and swaps as follows.
  Before each split,
  swap enough to break as few bonds as possible.
  After each merge,
  make as many bonds as possible.
  \qedhere
\end{proof}

\noindent
A saturated path simply needs no makes or breaks,
so this proof also proves the inequality for saturated paths
in the two models.

\begin{cor}
  $\satbondbarrier{\start}{\finish} + 1
  \geq \satbarrier{\simp{\start}}{\simp{\finish}}$.
\end{cor}

\endgroup

%% file: conclusion.tex
\section{Future work} \seclabel{conc}

An important type of kinetic barrier is the barrier to nucleation in self-assembly processes.
For example, in the abstract Tile Assembly Model (aTAM)~\cite{DotCACM,PatitzSurveyJournal}
a large nucleation barrier is necessary to ensure correct assembly of complex structures.
Prior work demonstrated programmable nucleation barriers, both theoretically and experimentally~\cite{SchWin09,SchWin07,minev2019crisscross}.
The TBN model could be used to develop aTAM kinetic barriers that do not rely on the tile geometry.
These barriers would be stronger, in the sense that they would exist even allowing lattice errors in assembled structures.

Ideally we not only want a large energy barrier to ``bad'' configurations, 
but we want to avoid getting stuck in local minima that keep us from getting to the good configurations.
We can define ``self-stabilizing'' TBNs with the property that from \emph{any} configuration, the TBN can reach \emph{some} stable configuration with a low energy barrier path.
This property is true for the grid gate (Theorem~\ref{self_stabilize_to_base}), but it is not true for the translator cycle. 
Is there a systematic way to test for self-stabilization, or ensure that a system is self-stabilizing by satisfying certain general properties?

Most of this paper considers the regime where forming a new bond is favorable even if it results a loss of a separate polymer (specifically, $w \geq 1$ and $w \geq 2$). 
However, we conjecture that both the translator cycle and the grid gate have $\Omega(n)$ energy barriers even when bond strength is weak ($w<1$) as long as $w = \Omega(1/n)$.
Showing that our constructions work in a wider range of experimental conditions would increase their practical applicability.  

Can we use the definition of energy in TBNs to bootstrap a reasonable notion of probability of configurations or paths?
For instance, 
in statistical thermodynamics it is common to consider the Boltzmann distribution induced by energy $E$,
where for each configuration $\gamma$,
$\Pr[\gamma] = e^{-E(\gamma)} / (\sum_{\beta} e^{-E(\beta)})$.
This is the probability of seeing $\gamma$ at thermodynamic equilibrium.
One can also use the relative energy between two states to predict the relative rates of transition between them, which might allow defining a notion of path probability in the kinetic theory of TBNs.

A useful chemical module consists of the reaction $X+X \to Y+Y$ (or more generally converting $k > 1$ copies of $X$ to $Y$), which can act as a ``threshold'' to  detect whether there are at least $k$ copies of $X$.
Analogous to a catalytic system, implementing the above reaction while forbidding $X \to Y$ requires control of the energy barrier, and cannot be done simply by varying the energies of $X$ and $Y$.
Can we construct TBNs with arbitrarily large energy barriers in this case?

What is the computational complexity of 
deciding whether $b(\gamma,\delta) \leq k$, 
given two configurations $\gamma,\delta$ and a threshold $k$?
This problem is decidable in polynomial space in the number of monomers in $\gamma$: any configuration can be written down in polynomial space, and guessing merges and splits yields a nondeterministic polynomial space algorithm (placing it in $\mathsf{PSPACE}$ by Savitch's Theorem).
However, low-height paths could be of exponential length,
and thus it is not clear that the problem is in $\mathsf{NP}$,
since the obvious witness is a path with height $\leq k$.
Is this problem possibly $\mathsf{PSPACE}$-complete?

Both the grid gate and translator cycle use $n^2$ unique site types to achieve an energy barrier of $n$.
This can be reduced for the grid gate
(for example, to use $n$ domains,
one can simply make each ``row'' of the grid the same site type), 
although we do not know how to make such a system work properly with a catalyst.
Specifically, Theorem~\ref{thm:copy_tolerant_grid_catalyst_stability} fails with our initial attempts to create a catalyst for systems with $O(n)$ site types.
It is an interesting open question to 
design a catalytic system similarly to the grid gate,
using only $O(n)$ site types,
that has a programmable energy barrier of $n$ in the absence of the catalyst.

%% file: translator_appendix.tex
\section{Translator cycle}

\subsection{Strand displacement cascade}\label{app:subsec:strand_displacement}
\figref{translator_dna} shows the design scheme in which the monomers/polymers of the TBN studied in Section~\ref{subsec:translator} are modeling an implementable substrate.
The substrate are double-stranded DNA complexes.
The complexes undergo a kinetic process called \emph{strand displacement}, shown in \figref{translator_displacement}, in which one strand attaches and displaces another, freeing a new strand which can then displace strands on other complexes in the system (hence the term \emph{cascade}).
Note that the displacement shown in \figref{translator_displacement} corresponds to one merge and one split of the pathway shown in TBN form in \figref{translator_ell_barrier}.

\begin{figure}[t]
  \centering
  \includegraphics[width=1\columnwidth]{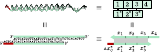}\
  \caption{DNA strand displacement implementation of the translator cycle. The schematic shows how the complexes in the translator cycle TBN correspond to DNA complexes.
  $\Delta x_i$ denotes a truncated version of domain $x_i$.
  Note that these truncated domains, typically called \emph{toeholds}, are not included in our TBN abstraction because they bind weakly and do not contribute to the energy barrier we consider.
  }
  \figlabel{translator_dna}
\end{figure}

\begin{figure}[t]
  \centering
  \includegraphics[width=1\columnwidth]{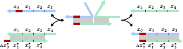}\
  \caption{
  The desired catalyzed pathway of the translator cycle consists of the strand displacement reactions shown. 
  The displacement reaction initiates when the blue strand binds to the green complex at the toehold domain (red). 
  The top strands compete for bonds in a random walk, eventually displacing the top green strand.
  Note that this displacement corresponds to one merge then split of the pathway shown in TBN form in \figref{translator_ell_barrier}.}
  \figlabel{translator_displacement}
\end{figure}

\begin{figure}[t]
  \centering
  \opt{ieee}{\includegraphics[width=.74\columnwidth]{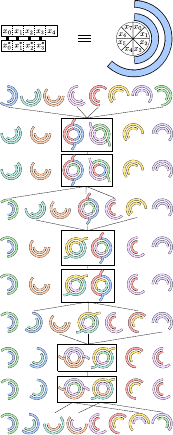}}
  \opt{llncs}{\includegraphics[height=.75\textheight]{translator_fl_path_circle.pdf}}
  \caption{
  A path from $\gamma_I$ to $\gamma_T$ with height $\frac{2\fuelcount}{\strandlength}-1$, showing that the complex length $\strandlength$ must be asymptotically larger than the number of complex types $\fuelcount$ to achieve a super-constant barrier.
  This example has monomer size $\strandlength = 4$ and number of complexes $\fuelcount = 8$.
  The circular representation helps to visualize when a group of monomers covers all site types $x_0,\ldots,x_7$.
  Each row of polymers denotes a different configuration on the path.
  Each is separated by several merges or splits.
  First, two initial complexes are brought together, chosen such that every starred binding site appears at least once.
  This size $\frac{2\fuelcount}{\strandlength}$ polymer acts as a catalyst in the following sense: another similar set of initial complexes can be merged into this polymer, and then split into their triggered complexes.
  Note in the fourth row, the large polymer is still saturated, although each bottom monomer is not in a polymer with its initial or triggered top monomer.}
  \figlabel{translator_f_barrier}
\end{figure}

\subsection{Barrier of the translator cycle}\label{app:translator}
Here we give proofs which were omitted in the main body.\\

\newcommand{\appendixref}[3]{\noindent\textbf{#1 #2. }\emph{#3}}

\appendixref{Lemma}{\ref{lem:perfect_matching}}{If a polymer $P$ is in a normal form saturated configuration and $|P| < 2n+1$, then there exists a perfect matching on $G_P$.}
\begin{proof}
Given a set $S$ of vertices, let $N(S)$ be the set of vertices adjacent to a vertex in $S$.
Hall's condition states that a perfect matching exists on a bipartite graph $\{V_1,V_2,E\}$ if and only if for all subsets $S \subseteq V_1$,
$|S| \leq |N(S)|$.
We prove this holds for $G_P$.

Consider any subset $S \subseteq B$.
There are $n|S|$ starred (limiting) binding sites.
The number of sites on compatible top monomers for the set $S$ is given by $(n+1)|N(S)|$.
Since $P$ is saturated, the $n|S|$ starred sites must be bound to the $(n+1)|N(S)|$ unstarred sites, so we have $n|S| \leq (n+1)|N(S)|$. If $|S| > |N(S)|$, Lemma~\ref{clm:xyn} gives us that $|S| > n+1$ and $|N(S)| > n$, so to avoid contradicting the assumption that $|P| < 2n+1$, it must be that $|S| \leq |N(S)|$.
\qedhere\end{proof}

\appendixref{Lemma}{\ref{clm:matchings_equal}}{For an $n'$-sized polymer $P$ in a normal form saturated configuration, for any two perfect matchings $M$ and $M'$ on $G_P$, $f(M) = f(M')$.}
\begin{proof}
Intuitively, first we will shift the indices of the top and bottom monomers so that the leftmost index has value zero.
Let $c_P$ be a cutoff value given by Lemma~\ref{clm:cutoff}.
\newcommand{\shift}{c_P}%
We rewrite each $b_i$ or $t_i$ as $b_{i - \shift \mod n^2}$ or $t_{i - \shift \mod n^2}$.
Since originally no edge crossed the cutoff, now no edge crosses zero.
Note that this does not change the offset of any pair and thus does not change the offset of any matching.
For each $f(b_i, t_j)$, since no edge $\{b_i, t_j\}$ crosses zero we can think of the indices on a line, so we can rewrite the offset as $f(b_i, t_j) = j-i$. Then $f(M) = \sum_k (j_k - i_k) = \sum_k j_k - \sum_k i_k$.
This expression is independent of the matching and only depends on the indices of the monomers in the polymer, so for any two matchings $M$ and $M'$, $f(M)=f(M')$.
\qedhere\end{proof}

\newcommand{\smod}[3]{[#1, #2]_{#3}}

\appendixref{Lemma}{\ref{clm:sorted}}{Given an $n'$-sized polymer $P$ with cutoff $c_P$, there exists a \emph{sorted} matching $M$ on $G_P$ which satisfies that there does not exist $\{b_{i_1}, t_{j_2}\},\{b_{i_2}, t_{j_1}\} \in M$ with $i_1 \leq i_2$ and $j_1 \leq j_2$ with respect to the ordering given by the cutoff value, $\smod{c_P}{c_P-1}{n^2}$.}
\begin{proof}
Given any matching $M'$ which is not sorted, we show that we can swap the offending edges, resulting in a new matching which is in sorted order.
Let $S_{i_k}$ be $\smod{i_k - n}{i_k+n-1}{n^2}$, the sequence giving the indices of compatible top monomers for $b_{i_k}$.
For any $\{b_{i_1}, t_{j_2}\},\{b_{i_2}, t_{j_1}\} \in M$ with $i_1 \leq i_2$ and $j_1 \leq j_2$, note that $j_2 \in S_{i_1}$ and $j_1 \in S_{i_2}$.
The orderings given by the sequences $S_{i_1}$, $S_{i_2}$ are the same orderings as given by $\smod{c_P}{c_P-1}{n^2}$, since $S_{i_1}$, $S_{i_2}$ are both subsequences of $\smod{c_P}{c_P-1}{n^2}$.
Since $i_1 \leq i_2$, $S_{i_1}$ contains no elements greater than any of those in $S_{i_2}$, and $S_{i_2}$ contains no elements less than any of those in $S_{i_1}$.
So $j_2 \in S_{i_1}$ and $j_1 \in S_{i_2}$ with $j_1 \leq j_2$ gives that both $j_1, j_2 \in S_{i_1}$ and further $j_1, j_2 \in S_{i_2}$.
So both $t_{j_1}$ and $t_{j_2}$ are compatible for both $b_{i_1}$ and $b_{i_2}$.
Since the edges of $G_P$ are between bottom monomers and their compatible tops, we can swap $\{b_{i_1}, t_{j_2}\},\{b_{i_2}, t_{j_1}\}$ with $\{b_{i_1}, t_{j_1}\},\{b_{i_2}, t_{j_2}\}$ and the result is a matching on $G_P$.
So given any perfect matching on $G_P$, we can construct a sorted matching by swapping the offending edges one-by-one.
\qedhere\end{proof}

%% file: biography.tex
\begin{IEEEbiographynophoto}{Keenan Breik}
is a student at the University of Texas at Austin.
\end{IEEEbiographynophoto}

\begin{IEEEbiographynophoto}{Cameron Chalk}
received his B.S. and M.S. degrees from the University of Texas Rio Grande Valley, Edinburg, TX, USA, in 2015 and 2017 respectively. He is currently pursuing a Ph.D. degree from the University of Texas at Austin, TX, USA.
\end{IEEEbiographynophoto}

\begin{IEEEbiography}[{\includegraphics[width=1in,height=1.25in,clip,keepaspectratio]{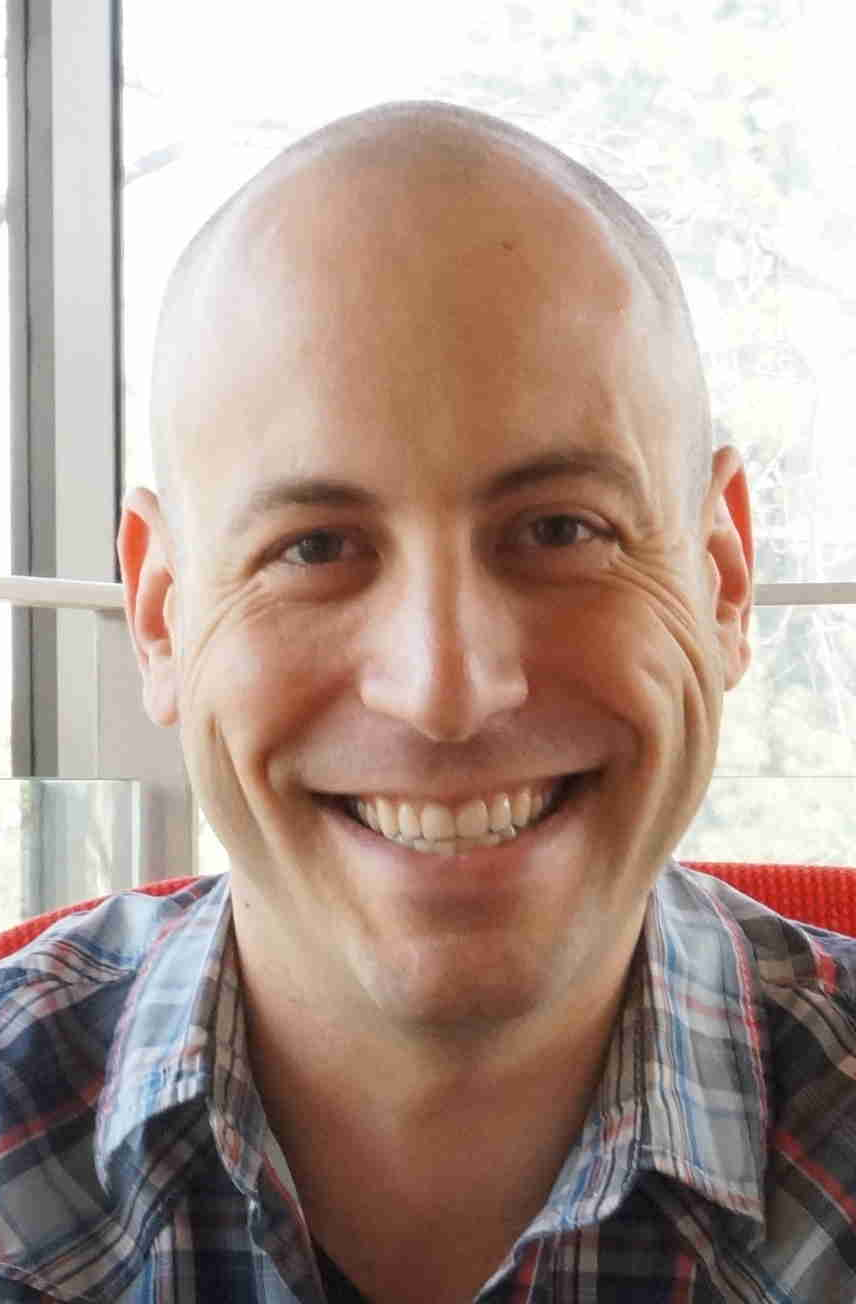}}]{David Doty}
is an Assistant Professor of Computer Science with the University of California at Davis.
\end{IEEEbiography}

\begin{IEEEbiography}[{\includegraphics[width=1in,height=1.25in,clip,keepaspectratio]{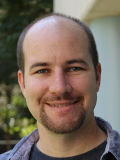}}]{David Haley}
is a Ph.D. candidate in the Graduate Group in Applied Mathematics at the University of California at Davis.
\end{IEEEbiography}

\begin{IEEEbiography}[{\includegraphics[width=1in,height=1.25in,clip,keepaspectratio]{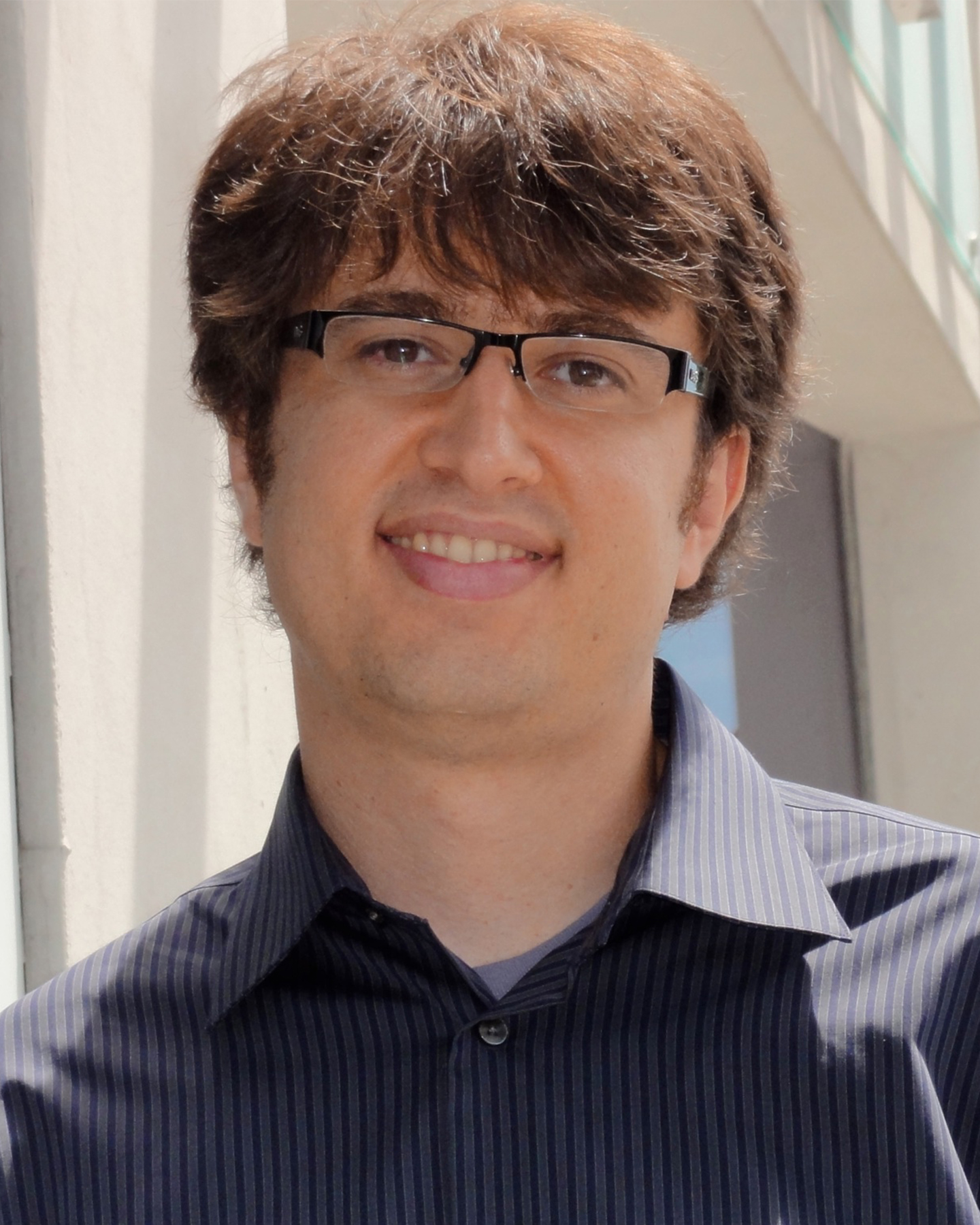}}]{David Soloveichik}
 received the B.S. and M.S. degrees from Harvard University, Cambridge, MA, USA, in 2002 and the Ph.D. degree from the California Institute of Technology, Pasadena, CA, USA, in 2008. He is currently an Assistant Professor in the Department of Electrical and Computer Engineering at the University of Texas at Austin, TX, USA.
\end{IEEEbiography}

\vfill

%% file: ms.bbl
\providecommand{\noopsort}[1]{}
\begin{thebibliography}{10}

\bibitem{baccouche2014dynamic}
Alexandre Baccouche, Kevin Montagne, Adrien Padirac, Teruo Fujii, and Yannick
  Rondelez.
\newblock Dynamic {DNA}-toolbox reaction circuits: a walkthrough.
\newblock {\em Methods}, 67(2):234--249, 2014.

\bibitem{breik2019computing}
Keenan Breik, Chris Thachuk, Marijn Heule, and David Soloveichik.
\newblock Computing properties of stable configurations of thermodynamic
  binding networks.
\newblock {\em Theoretical Computer Science}, 785:17--29, 2019.

\bibitem{cardelli2011strand}
Luca Cardelli.
\newblock Strand algebras for {DNA} computing.
\newblock {\em Natural Computing}, 10(1):407--428, 2011.

\bibitem{chen2013programmable}
Yuan-Jyue Chen, Neil Dalchau, Niranjan Srinivas, Andrew Phillips, Luca
  Cardelli, David Soloveichik, and Georg Seelig.
\newblock Programmable chemical controllers made from {DNA}.
\newblock {\em Nature Nanotechnology}, 8(10):755--762, 2013.

\bibitem{DotCACM}
David Doty.
\newblock Theory of algorithmic self-assembly.
\newblock {\em Communications of the ACM}, 55(12):78--88, December 2012.

\bibitem{tbn}
David Doty, Trent~A. Rogers, David Soloveichik, Chris Thachuk, and Damien
  Woods.
\newblock Thermodynamic binding networks.
\newblock In Robert Brijder and Lulu Qian, editors, {\em DNA Computing and
  Molecular Programming: 23rd International Conference}, pages 249--266.
  Springer, 2017.

\bibitem{minev2019crisscross}
Dionis Minev, Christopher Wintersinger, and William~M Shih.
\newblock Crisscross cooperative self-assembly, July~4 2019.
\newblock US Patent App. 16/322,787.

\bibitem{PatitzSurveyJournal}
Matthew~J. Patitz.
\newblock An introduction to tile-based self-assembly and a survey of recent
  results.
\newblock {\em Natural Computing}, 13(2):195--224, 2014.

\bibitem{SchWin07}
Rebecca Schulman and Erik Winfree.
\newblock Synthesis of crystals with a programmable kinetic barrier to
  nucleation.
\newblock {\em Proceedings of the National Academy of Sciences},
  104(39):15236--15241, 2007.

\bibitem{SchWin09}
Rebecca Schulman and Erik Winfree.
\newblock Programmable control of nucleation for algorithmic self-assembly.
\newblock {\em SIAM Journal on Computing}, 39(4):1581--1616, 2009.
\newblock Preliminary version appeared in DNA 2004.

\bibitem{soloveichik2010dna}
David Soloveichik, Georg Seelig, and Erik Winfree.
\newblock {DNA} as a universal substrate for chemical kinetics.
\newblock {\em Proceedings of the National Academy of Sciences},
  107(12):5393--5398, 2010.

\bibitem{srinivas2017enzyme}
Niranjan Srinivas, James Parkin, Georg Seelig, Erik Winfree, and David
  Soloveichik.
\newblock Enzyme-free nucleic acid dynamical systems.
\newblock {\em Science}, 358(6369):eaal2052, 2017.

\bibitem{thachuk2015leakless}
Chris Thachuk, Erik Winfree, and David Soloveichik.
\newblock Leakless {DNA} strand displacement systems.
\newblock In {\em International Workshop on DNA-Based Computers}, pages
  133--153. Springer, 2015.

\bibitem{wang2018}
Boya Wang, Chris Thachuk, Andrew~D Ellington, Erik Winfree, and David
  Soloveichik.
\newblock Effective design principles for leakless strand displacement systems.
\newblock {\em Proceedings of the National Academy of Sciences},
  115(52):E12182--E12191, 2018.

\bibitem{yin2008programming}
Peng Yin, Harry~MT Choi, Colby~R Calvert, and Niles~A Pierce.
\newblock Programming biomolecular self-assembly pathways.
\newblock {\em Nature}, 451(7176):318, 2008.

\bibitem{ZhaTurYurWin07}
David~Yu Zhang, Andrew~J. Turberfield, Bernard Yurke, and Erik Winfree.
\newblock Engineering entropy-driven reactions and networks catalyzed by {DNA}.
\newblock {\em Science}, 318(5853):1121--1125, 2007.

\end{thebibliography}
